\documentclass[11pt]{article}

\usepackage{amsmath, fullpage}
\usepackage{xspace}
\usepackage{amssymb}
\usepackage{epsfig}

\newtheorem{Thm}{Theorem}
\newtheorem{Lem}{Lemma}
\newtheorem{Cor}{Corollary}

\newtheorem{Claim}{Claim}

\newtheorem{Def}{Definition}
\newtheorem{Fact}{Fact}
\newenvironment{proof}{\noindent {\textbf{Proof }}}{$\Box$
\medskip}

\newcommand{\defeq}{\stackrel{\mathsf{def}}{=}}

\newcommand\mbR{\mbox{$\mathbb{R}$}}
\newcommand\mcX{\mathcal{X}}
\newcommand\mcY{\mathcal{Y}}
\newcommand\mcZ{\mathcal{Z}}
\newcommand\mcA{\mathcal{A}}

\newcommand\mcP{\mathcal{P}}
\newcommand\mcR{\mathcal{R}}

\newcommand{\alice}{{\mathsf{Alice}}}
\newcommand{\bob}{{\mathsf{Bob}}}
\newcommand{\prt}{{\mathsf{prt}}}
\newcommand{\rec}{{\mathsf{rec}}}
\newcommand{\trec}{{\widetilde{\mathsf{rec}}}}
\newcommand{\srec}{{\mathsf{srec}}}
\newcommand{\tsrec}{{\widetilde{\mathsf{srec}}}}
\newcommand{\disc}{{\mathsf{disc}}}
\newcommand{\sdisc}{{\mathsf{sdisc}}}
\newcommand{\tsdisc}{{\widetilde{\mathsf{sdisc}}}}
\newcommand{\adv}{{\mathsf{adv}}}
\newcommand{\cadv}{{\mathsf{cadv}}}
\newcommand{\D}{{\mathsf{D}}}
\newcommand{\R}{{\mathsf{R}}}
\newcommand{\C}{{\mathsf{C}}}
\newcommand{\s}{{\mathsf{s}}}
\newcommand{\RC}{{\mathsf{RC}}}
\newcommand{\bs}{{\mathsf{bs}}}
\newcommand{\pub}{{\mathsf{pub}}}
\newcommand{\tdeg}{{\widetilde{\mathsf{deg}}}}
\newcommand{\edeg}{{{\mathsf{deg}}}}

\newcommand\disjointness{{\mathsf{Disjointness}}}
\newcommand\disj{{\mathsf{Disj}}}
\newcommand\LNE{\mathsf{LNE}}
\newcommand\tribes{{\mathsf{Tribes}}}
\newcommand\rank{{\mathsf{rank}}}
\newcommand\poly{{\mathsf{poly}}}

\mathchardef\mhyphen="2D
\newcommand{\ip}[2]{\left\langle #1 , #2\right\rangle}
\newcommand{\signorm}[1]{||#1||_\Sigma}

\newcommand{\suppress}[1]{}
\newcommand\COMMENT[1]{}

\newcommand\norm[1]{\| #1 \|}
\newcommand\abs[1]{| #1 |}

\begin{document}
\title{The Partition Bound for Classical Communication Complexity and Query Complexity}
\author{Rahul Jain\thanks{Centre for Quantum Technologies and Department of Computer Science, National University of Singapore. Email: {\tt rahul@comp.nus.edu.sg}} \quad Hartmut Klauck\thanks{Centre for Quantum Technologies, National University of Singapore. Email: {\tt hklauck@gmail.com}} }
\date{}
\maketitle

\begin{abstract}
We describe new lower bounds for randomized communication
complexity and query complexity which we call the {\em partition
bounds}. They are expressed as the optimum value of linear programs. For communication complexity we
show that the partition bound is stronger than both the
{\em rectangle/corruption bound} and the {\em $\gamma_2$/generalized discrepancy
bounds}. In the model of query complexity we show that the partition
bound is stronger than the {\em approximate
polynomial degree} and {\em classical adversary bounds}. We also exhibit an
example where the partition bound is quadratically larger than the approximate polynomial degree and adversary bounds.
\end{abstract}


\section{Introduction}

The computational models investigated in communication complexity and
query complexity, i.e., Yao's communication model \cite{yao:cc} and
the decision tree model, are simple enough to allow us to prove
interesting lower bounds, yet they are rich enough to have numerous
applications to other models as well as exhibit nontrivial structure. Research
in both these models is concentrated on lower bounds and a recurring theme
is methods to prove such bounds. In this paper we present a new method
for proving lower bounds on randomized complexity in both of these models.

\subsection{Communication Complexity}

In the model of communication complexity there are several general
methods to prove lower bounds in the settings of randomized
communication and quantum communication. Linial and Shraibman
\cite{linial:norms} identified a quantity called $\gamma_2$, which not
only yields lower bounds for quantum protocols, but also subsumes a
good number of previously known bounds. Later, Sherstov
\cite{sherstov:pattern} described a quantity called generalized
discrepancy (the name being coined in \cite{chatto:multiparty}), which
also coincides with $\gamma_2$. The generalized discrepancy can be
derived from the standard discrepancy bound (see
\cite{kushilevitz&nisan:cc}, this bound was shown to be applicable in
the quantum case by Kremer and Yao \cite{kremer:thesis}) in a way
originally suggested by Klauck \cite{klauck:lbqcc}. In particular,
Sherstov showed that the $\gamma_2$ method yields a tight
$\Omega(\sqrt n)$ bound for the quantum communication complexity of
the $\disjointness$ problem, arguably the most important single function
considered in the area. This result was previously established by a
more complicated method \cite{razborov:qdisj}, for a matching upper
bound see \cite{aaronson:search}. This leaves our knowledge of lower
bound methods in the world of quantum communication complexity in a
neat form, where there is one "master method" that seems to do better
than everything else; the only potential competition coming from {\em information
theoretic techniques} like in \cite{jayram:app,jain:info}, which are
not applicable to all problems, and not known to beat $\gamma_2$
either.

In the world of randomized communication things appear to be much less
organized. Besides simply applying $\gamma_2$, the main competitors
are the rectangle (aka corruption) bound (compare
\cite{yao:prob, bfs:classes,razborov:disj, klauck:thresh, beame:sdpt}), as well
as again information theoretic techniques. Both of the latter
approaches are able to beat $\gamma_2$, by allowing $\Omega(n)$ bounds
for the $\disjointness$ problem \cite{razborov:qdisj, bar-yossef:disj,  ks:disj},
and there is an information theoretic proof of a tight $\Omega(n)$
lower bound for the $\tribes$ function (an AND of $\sqrt n$ ORs of $\sqrt
n$ ANDs of distributed pairs of variables \cite{jayram:app}). With the
rectangle bound one cannot prove a lower bound larger than $\sqrt n$
for this problem, and neither with $\gamma_2$. So the two general
techniques, rectangle bound and $\gamma_2$, are known to be
quadratically smaller than the randomized communication complexity for
some problems, and the information theoretic approach seems to be only
applicable to problems of a "direct sum" type.

In this paper we propose a new lower bound method for randomized
communication complexity which we call the partition bound \footnote{In this paper we are only concerned with the two-party model and the partition bound for other models can be defined analogously. For example for the {\em Number on the Forehead Model} it can be defined by replacing rectangles by {\em cylinder intersections}.}. We derive this
bound from a linear program, which captures a relaxation of the
fact that a randomized protocol is a convex combination of
deterministic protocols and hence a convex combination of partitions
of the communication matrix into rectangles. Linear programs have been
previously used to describe lower bounds in communication complexity. Lovasz \cite{lovasz:cc} gives a program which, as we show, turns out to capture the rectangle bound. Our program for the partition bound however uses stricter constraints to overcome the one-sidedness of Lovasz's program. Karchmer et al.~\cite{karchmer:fractional} give a linear program for {\em fractional covers}, as well as a linear program which can be seen to be equivalent to our zero-error partition bound for relations, where it was introduced as a lower bound to deterministic complexity.

We also describe a weaker bound to the partition bound which we call
the "smooth rectangle bound". It is inherently a one-sided bound
and is derived by relaxing constraints in the linear program for the partition bound.  This bound
has recently been used to prove a {\em strong direct product theorem}
for $\disjointness$ in \cite{klauck:disjointness}. Another way to derive
the smooth rectangle bound is as follows. Suppose we want to prove a
lower bound for a function $f$. Then we could apply the rectangle
bound, but sometimes this might not yield a large enough lower
bound. Instead we apply the rectangle bound to a function $g$ that is
sufficiently close to $f$ (under a suitable probability distribution), so that lower bounds for $g$ imply lower
bounds for $f$. Maximizing this over
all $g$, close to $f$, gives us the smooth rectangle bound. This is the same approach
that turns the discrepancy bound into the generalized discrepancy (see
\cite{sherstov:pattern,klauck:lbqcc}). We will use the term smooth
discrepancy in the following, because it better captures the
underlying approach.

After defining the partition bound and the smooth rectangle bound we
proceed to show that the smooth rectangle bound subsumes both the standard
rectangle bound and $\gamma_2$/smooth discrepancy. We also show that
the LP formulation of the smooth rectangle bound coincides with its
natural definition as described above. This leaves us with one unified
general lower bound method for randomized communication complexity,
the partition bound.

\subsection{Query Complexity}
We then turn to randomized query complexity. Again there are several prominent lower bound
methods in this area. Some of the main methods are the classical version of Ambainis' adversary
method (the quantum version is from \cite{ambainis:adversary}, and
classical versions are  by Laplante/Magniez \cite{laplante:kolmogorov}
and Aaronson \cite{aaronson:local});  the approximate polynomial degree
\cite{nisan&szegedy:degree,bbcmw:polynomials}; the randomized
certificate bound defined by Aaronson \cite{aaronson:cert} (this being
the query complexity analogue of the rectangle bound in communication
complexity), as well as older methods like block-sensitivity
\cite{nisan:pram}.

We again propose a new lower bound, the partition bound, defined via a
linear program, this time based on the fact that a decision tree partitions the Boolean
cube into subcubes. We then proceed to show that our lower bound
method subsumes all the other bounds mentioned above. In particular
the partition bound is always larger than the classical adversary
bound, the approximate degree, and block-sensitivity.

To further illustrate the power of our approach we describe a Boolean
function, (AND of ORs), which we continue to call  $\tribes$, for which the partition
bound yields a tight linear lower bound, while both the adversary
bound and the approximate degree are at least quadratically
smaller.

\section{Communication Complexity Bounds}
In this section we present the definition of the partition bound and the smooth-rectangle bound followed by the definitions of the previously known lower bounds for randomized communication complexity. Subsequently, in the next subsection, we present key relationships and comparisons between various bounds.
\subsection{Definitions}
Let $f : \mcX \times \mcY \rightarrow \mcZ$ be a partial function. All
the functions considered in this section are partial functions unless
otherwise specified, hence we will drop the term partial henceforth.
It is easily verified that strong duality holds for the programs that
appear below and hence optima for the primal and dual are same. Let
$\mcR$ be the set of all rectangles in $\mcX \times \mcY$. We refer
the reader to~\cite{kushilevitz&nisan:cc} for introduction to basic terms in
communication complexity. Below we
assume $(x,y) \in \mcX \times \mcY, R \in \mcR, z \in \mcZ$, unless
otherwise specified. Let $f^{-1} \subseteq \mcX\times \mcY$ denote the subset where $f(\cdot)$ is defined. For sets $A,B$ we denote $A-B \defeq \{a: a\in A, a \notin B\}$. We assume $\epsilon \geq 0 $ unless otherwise specified.
\subsubsection{New Bounds}
\begin{Def}[Partition Bound]
The $\epsilon$-partition bound of $f$, denoted $\prt_\epsilon(f)$, is given by the optimal value of the following linear program.

\vspace{0.2in}

{\footnotesize
\begin{minipage}{3in}
    \centerline{ \underline{Primal}}\vspace{-0.2in}
 \begin{align*}
      \text{min:}\quad & \sum_{z} \sum_{R}  w_{z,R} \\
      \quad & \forall (x,y) \in f^{-1}: \sum_{R: (x,y) \in R} w_{f(x,y),R} \geq 1 - \epsilon,\\
      & \forall (x,y) : \sum_{R: (x,y) \in R} \quad \sum_{z}  w_{z,R} = 1 , \\
      & \forall z , \forall R  : w_{z,R} \geq 0 \enspace .
    \end{align*}
\end{minipage}
\begin{minipage}{3in}\vspace{-0.35in}
    \centerline{\underline{Dual}}\vspace{-0.2in}
 \begin{align*}
      \text{max:}\quad &  \sum_{(x,y)\in f^{-1}} (1-\epsilon) \mu_{x,y} + \sum_{(x,y)}\phi_{x,y}\\
      \quad &  \forall z , \forall R : \sum_{(x,y)\in f^{-1}(z)\cap R} \mu_{x,y}   + \sum_{(x,y)\in  R} \phi_{x,y}  \leq 1,\\
      & \forall (x,y) : \mu_{x,y} \geq 0, \phi_{x,y} \in \mbR \enspace .
 \end{align*}
\end{minipage}
}
\end{Def}

Below we present the definition of smooth-rectangle bound as a one-sided relaxation of the partition bound. As we show in the next subsection, it is upper bounded by the partition bound.

\begin{Def}[Smooth-Rectangle bound]
The $\epsilon$- smooth rectangle bound of $f$ denoted $\srec_\epsilon(f)$ is defined to be $\max\{\srec^z_\epsilon(f): z\in\mcZ\}$, where $\srec^z_\epsilon(f)$ is given by the optimal value of the following linear program.

\vspace{0.2in}

{\footnotesize
\hspace{-0.4in}\begin{minipage}{3in}
    \centerline{\underline{Primal}}
    \begin{align*}
      & \text{min:}\quad  \sum_{R \in \mcR}  w_{R} \\
       \quad &  \forall (x,y) \in  f^{-1}(z): \sum_{R: (x,y) \in R} w_{R} \geq 1 - \epsilon,\\
      & \forall (x,y) \in  f^{-1}(z): \sum_{R: (x,y) \in R} w_{R} \leq 1,\\
      &  \forall (x,y) \in  f^{-1} - f^{-1}(z): \sum_{R: (x,y) \in R} w_{R} \leq \epsilon,\\
      & \forall R : w_{R} \geq 0 \enspace .
    \end{align*}
\end{minipage}
\begin{minipage}{3in}\vspace{-0.7in}
    \centerline{\underline{Dual}}
    \begin{align*}
      & \text{max:}\quad   \sum_{(x,y)\in f^{-1}(z)} \left((1-\epsilon) \mu_{x,y} - \phi_{x,y} \right)- \sum_{(x,y)\in f^{-1} - f^{-1}(z)} \epsilon \cdot \mu_{x,y}\\
       \quad &   \forall R : \sum_{(x,y)\in f^{-1}(z)\cap R} (\mu_{x,y} - \phi_{x,y}) - \sum_{(x,y)\in (R \cap f^{-1}) - f^{-1}(z)} \mu_{x,y} \leq 1,\\
      & \forall (x,y) : \mu_{x,y} \geq 0 ; \phi_{x,y} \geq 0 \enspace .
    \end{align*}
\end{minipage}
}
\end{Def}

Below we present an alternate and "natural" definition of smooth-rectangle bound, which justifies its name. In the next subsection we show that the two definitions are equivalent.

\begin{Def}[Smooth-Rectangle bound : Natural definition] In the natural definition, $(\epsilon,\delta)$- smooth-rectangle bound of $f$,  denoted $\tsrec_{\epsilon,\delta}(f) $, is defined as follows (refer to the definition of $\trec^{z,\lambda}_\epsilon(g)$ in the next subsection):
\begin{align*}
& \tsrec_{\epsilon, \delta}(f)  \defeq \max\{ \tsrec^z_{\epsilon, \delta}(f) : z\in \mcZ\}.  \\
& \tsrec^z_{\epsilon,\delta}(f) \defeq \max\{\tsrec^{z,\lambda}_{\epsilon, \delta}(f) : \lambda \mbox{ a (probability) distribution on } \mcX \times \mcY \cap f^{-1} \}.\\
& \tsrec^{z,\lambda}_{\epsilon,\delta}(f) \defeq \max\{ \trec^{z,\lambda}_\epsilon(g) : g: \mcX \times \mcY \rightarrow \mcZ; \Pr_{(x,y) \leftarrow \lambda}[f(x,y) \neq g(x,y)] < \delta;  \lambda(g^{-1}(z)) \geq 0.5\} .
\end{align*}
\end{Def}
Below we  define smooth-discrepancy via a linear program. In the next subsection we present the natural definition of smooth-discrepancy and in the subsequent subsection we show that the two definitions are equivalent. As we also show in the next subsection smooth-discrepancy is upper bounded by smooth-rectangle bound which in turn is upper bounded by the partition bound.
\begin{Def}[Smooth-Discrepancy] Let $f :\mcX \times \mcY \rightarrow \{0,1\}$ be a Boolean function. The smooth-discrepancy of $f$, denoted $\sdisc_\epsilon(f)$, is given by the optimal value of the following linear program.

\vspace{0.2in}

{\footnotesize
\hspace{-0.2in}\begin{minipage}{3in}
    \centerline{\underline{Primal}}
    \begin{align*}
      & \text{min:}\quad  \sum_{R \in \mcR}  w_{R} + v_R\\
       \quad &  \forall (x,y) \in  f^{-1}(1): \quad 1+ \epsilon \geq \sum_{R: (x,y) \in R} w_{R}  - v_R \geq 1 ,\\
      & \forall (x,y) \in  f^{-1}(0): \quad  1+ \epsilon \geq \sum_{R: (x,y) \in R} v_R - w_{R} \geq 1,\\
      & \forall R : w_{R}, v_R \geq 0 \enspace .
    \end{align*}
\end{minipage}
\begin{minipage}{3in}
    \centerline{\underline{Dual}}
    \begin{align*}
      & \text{max:}\quad   \sum_{(x,y)\in f^{-1}}  \mu_{x,y}  - (1+\epsilon)  \phi_{x,y}\\
       \quad &   \forall R : \sum_{(x,y)\in f^{-1}(1)\cap R} (\mu_{x,y} - \phi_{x,y}) - \sum_{(x,y)\in R \cap f^{-1}(0) } (\mu_{x,y} - \phi_{x,y}) \leq 1,\\
      &   \forall R : \sum_{(x,y)\in f^{-1}(0)\cap R} (\mu_{x,y} - \phi_{x,y}) - \sum_{(x,y)\in R \cap f^{-1}(1) } (\mu_{x,y} - \phi_{x,y}) \leq 1,\\
      & \forall (x,y) : \mu_{x,y} \geq 0 ; \phi_{x,y} \geq 0 \enspace .
    \end{align*}
\end{minipage}
}
\end{Def}

\subsubsection{Known Bounds}
Below we present the definition of the rectangle bound via a linear program. This program was first described by Lovasz~\cite{lovasz:cc} although he did not make the connection to the rectangle bound.
\begin{Def}[Rectangle-Bound]
The $\epsilon$-rectangle bound of $f$, denoted $\rec_\epsilon(f)$, is defined to be $\max\{\rec^z_\epsilon(f): z\in\mcZ\}$, where $\rec^z_\epsilon(f)$ is given by the optimal value of the following linear program.

\vspace{0.2in}

{\footnotesize
\hspace{-0.2in}\begin{minipage}{3in}
    \centerline{\underline{Primal}}
    \begin{align*}
      \text{min:}\quad & \sum_{R}  w_{R} \\
       \quad &  \forall (x,y) \in f^{-1}(z): \sum_{R: (x,y) \in R} w_{R} \geq 1 - \epsilon,\\
      &  \forall (x,y) \in  f^{-1} - f^{-1}(z): \sum_{R: (x,y) \in R} w_{R} \leq \epsilon,\\
      & \forall R  : w_{R} \geq 0 \enspace .
    \end{align*}
\end{minipage}
\begin{minipage}{3in}\vspace{-0.35in}
    \centerline{\underline{Dual}}
    \begin{align*}
      \text{max:}\quad &  \sum_{(x,y) \in f^{-1}(z)}  (1-\epsilon) \cdot  \mu_{x,y}  - \sum_{(x,y) \in f^{-1} - f^{-1}(z)} \epsilon \cdot \mu_{x,y}\\
       \quad &   \forall R : \sum_{(x,y)\in f^{-1}(z)\cap R} \mu_{x,y}  - \sum_{(x,y)\in (R \cap f^{-1}) - f^{-1}(z)} \mu_{x,y} \leq 1,\\
      & \forall (x,y) : \mu_{x,y} \geq 0 \enspace .
    \end{align*}
\end{minipage}
}
\end{Def}

Below we present the alternate, natural and conventional definition of rectangle bound as used in several previous works~\cite{yao:prob, bfs:classes,razborov:disj, klauck:thresh,  beame:sdpt}. In the next subsection we show that the two definitions are equivalent.
\begin{Def}[Rectangle-Bound: Conventional definition] In the conventional definition, $\epsilon$-rectangle bound of $f$, denoted $\trec_\epsilon(f)$  is defined as follows:
\begin{align*}
& \trec_\epsilon(f) \defeq \max\{\trec^z_\epsilon(f) : z\in\mcZ\} \\
& \trec^z_\epsilon(f) \defeq \max \{ \trec^{z,\lambda}_\epsilon(f) : \lambda \mbox{ a distribution on } \mcX \times \mcY \cap f^{-1} \mbox{ with } \lambda(f^{-1}(z)) \geq 0.5 \}. \\
& \trec^{z,\lambda}_\epsilon(f) \defeq \min\{ \frac{1}{\lambda(f^{-1}(z)\cap R )} : R \in \mcR \mbox{ with } \epsilon \cdot \lambda(f^{-1}(z)\cap R)  >  \lambda(R - f^{-1}(z)) \} \enspace .
\end{align*}
\end{Def}
Below we present the definition of discrepancy via a linear program followed by the conventional definition of discrepancy. It is easily seen that the two are exactly the same.
\begin{Def}[Discrepancy] Let $f :\mcX \times \mcY \rightarrow \{0,1\}$ be a Boolean function. The discrepancy of $f$, denoted $\disc(f)$, is given by the optimal value of the following linear program.

\vspace{0.2in}

{\footnotesize
\begin{minipage}{3in}
    \centerline{\underline{Primal}}\vspace{-0.2in}
    \begin{align*}
      \text{min:}\quad & \sum_{R}  w_{R} + v_R\\
       \quad &  \forall (x,y) \in  f^{-1}(1): \sum_{R: (x,y) \in R} w_{R}  - v_R \geq 1 ,\\
      & \forall (x,y) \in  f^{-1}(0): \sum_{R: (x,y) \in R} v_R - w_{R} \geq 1,\\
      & \forall R : w_{R}, v_R \geq 0 \enspace .
    \end{align*}
\end{minipage}
\begin{minipage}{3in}
    \centerline{\underline{Dual}}\vspace{-0.2in}
    \begin{align*}
      \text{max:}\quad &  \sum_{(x,y)\in f^{-1}} \mu_{x,y}  \\
       \quad &   \forall R : \sum_{(x,y)\in f^{-1}(1)\cap R} \mu_{x,y}  - \sum_{(x,y)\in R \cap f^{-1}(0) } \mu_{x,y} \leq 1,\\
      &   \forall R : \sum_{(x,y)\in f^{-1}(0)\cap R} \mu_{x,y}  - \sum_{(x,y)\in R \cap f^{-1}(1) } \mu_{x,y} \leq 1,\\
      & \forall (x,y) : \mu_{x,y} \geq 0\enspace .
    \end{align*}
\end{minipage}
}
\end{Def}

\begin{Def}[Discrepancy: Conventional definition] Let $f : \mcX \times \mcY \rightarrow \{0,1\}$ be a Boolean function. The discrepancy of $f$, denoted $\disc(f)$  is defined as follows:
\begin{align*}
& \disc(f) \defeq \max \{ \disc^{\lambda}(f) : \lambda \mbox{ a distribution on } \mcX \times \mcY \cap f^{-1}\}. \\
& \disc^{\lambda}(f) \defeq \min\{ \frac{1}{|\sum_{(x,y) \in R} (-1)^{f(x,y)} \cdot \lambda_{x,y}|} : R \in \mcR  \} \enspace .
\end{align*}
\end{Def}
Below we present the natural definition of smooth-discrepancy which has found shape in previous works~\cite{klauck:lbqcc, sherstov:pattern}. It is defined in analogous fashion from discrepancy as smooth-rectangle bound is defined from rectangle bound.
\begin{Def}[Smooth-Discrepancy: Natural Definition] Let $f : \mcX \times \mcY \rightarrow \{0,1\}$ be a Boolean function. The $\delta$- smooth-discrepancy of $f$,  denoted $\tsdisc_{\delta}(f) $, is defined as follows:
\begin{align*}
& \tsdisc_{\delta}(f) \defeq \max\{\tsdisc^{\lambda}_{\delta}(f) : \lambda \mbox{ a distribution on } \mcX \times \mcY \cap f^{-1}\}.\\
& \tsdisc^{\lambda}_{\delta}(f) \defeq \max\{\disc^{\lambda}(g) : g: \mcX \times \mcY \rightarrow \mcZ; \Pr_{(x,y) \leftarrow \lambda}[f(x,y) \neq g(x,y)] < \delta  \} .
\end{align*}
\end{Def}
Below we define the $\gamma_2$ bound of Linial and Shraibman~\cite{linial:norms} and show in the next subsection that it is equivalent to smooth-discrepancy.
\begin{Def}[$\gamma_2$ bound~\cite{linial:norms}] Let A be a sign matrix and $\alpha \geq 1$. Then,
$$ \gamma_2(A) \defeq \min_{X,Y : XY = A} r(X)c(Y) \quad ; \quad \gamma_2^\alpha(A) \defeq \min_{B: \forall (i,j) \; 1 \leq A(i,j)B(i,j) \leq \alpha} \gamma_2(B) .$$
Above $r(X)$ represents the largest $\ell_2$ norm of the rows of $X$ and $c(X)$ represents the largest $\ell_2$ norm of the columns of $Y$.
\end{Def}

Below we present two well-known lower bound methods for deterministic communication complexity.
\begin{Def}[$\log$-rank bound] Let $f: \mcX \times \mcY \rightarrow \mcZ$ be a total function. Let $M_f$ denote the communication matrix associated with $f$;  $\D(f)$ denote the deterministic communication complexity of $f$ and $\rank(\cdot)$ represents the rank over the reals. Then it is well known~\cite{kushilevitz&nisan:cc} that $\D(f) \geq \log_2 \rank(f)$.
\end{Def}
\begin{Def}[Fooling Set]Let $f: \mcX \times \mcY \rightarrow \mcZ$ be a total function. A set $S \subseteq \mcX \times \mcY$ is called a {\em fooling set} (for $f$) if there exists a value $z \in \mcZ$, such that
\begin{itemize}
\item For every $(x,y) \in S, f(x,y) = z$.
\item For every two distinct pairs $(x_1,y_1)$ and $(x_2,y_2)$ in $S$, either $f(x_1,y_2) \neq z$ or $f(x_2,y_1) \neq z$.
\end{itemize}
It is easily argued that $\D(f) \geq \log_2 |S|$~\cite{kushilevitz&nisan:cc}.
\end{Def}

\subsection{Comparison between bounds}
The following theorem captures key relationships between the bounds defined in the previous section. Below $\R^\pub_\epsilon(f)$ denotes the public-coin, $\epsilon$-error communication complexity of $f$.
\begin{Thm}Let $f : \mcX \times \mcY \rightarrow \mcZ$ be a function.
\begin{enumerate}
\item $\R^\pub_\epsilon(f) \geq \log \prt_\epsilon(f)$.
\item $ \prt_\epsilon(f) \geq  \srec_\epsilon(f)$.
\item $ \srec_\epsilon(f) \geq \rec_\epsilon(f)$.
\item Let $f: \mcX \times \mcY \rightarrow \mcZ$ be a total function, then $ \D(f) = O((\log \prt_0(f) + \log n)^2) $. Later we exhibit that the quadratic gap between $\D$ and $\log \prt_0$ is tight. For relations however there could be an exponential gap between $\log\prt_0$ and $\D$ as shown in~\cite{karchmer:fractional}.
\item Let $f: \mcX \times \mcY \rightarrow \mcZ$ be a total function, and let $S \subseteq \mcX \times \mcY$ be a fooling set. Then $ \prt_0(f) \geq |S| $.
\end{enumerate} \label{thm:main}
\end{Thm}
\begin{proof}
\begin{enumerate}
\item 
Let $\mcP$ be a public coin randomized protocol for $f$ with communication $c \defeq \R^\pub_\epsilon(f)$ and worst case error $\epsilon$. For binary string $r$, let $\mcP_r$ represent the deterministic protocol obtained from $\mcP$ on fixing the public coins to $r$. Let $r$ occur with probability $q(r)$ in $\mcP$. Every deterministic protocol amounts to partitioning the inputs in $\mcX\times\mcY$ into rectangles. Let $\mcR_r$ be the set of rectangles corresponding to different communication strings between $\alice$ and $\bob$ in $\mcP_r$. We know that $\abs{\mcR_r} \leq 2^c$, since the communication in $\mcP_r$ is at most $c$ bits. Let $z^r_R \in \mcZ$ be the output corresponding to rectangle  $R$ in $\mcP_r$. Let
$$w'_{z,R} \defeq  \sum_{r : R \in \mcR_r \mbox{ and } z^r_R = z} q(r) \enspace .$$
It is easily seen that for all $(x,y,z) \in \mcX \times\mcY\times\mcZ$:
$$\Pr[\mcP \mbox{ outputs } z \mbox{ on input } (x,y)] = \sum_{R: (x,y)\in R} w'_{z,R} \enspace .$$
Since the protocol has error at most $\epsilon$ on all inputs in $f^{-1}$ we get the constraints:
$$\forall (x,y) \in  f^{-1} : \sum_{R: (x,y) \in R} w'_{f(x),R} \geq 1 - \epsilon \enspace .$$
Also since the $\Pr[\mcP \mbox{ outputs some } z\in\mcZ \mbox{ on input } (x,y)]= 1$, we get the constraints:
$$\forall (x,y) :  \sum_{z}  \sum_{R: (x,y) \in R} w'_{z,R} = 1 \enspace .$$
Of course we also have by construction :
$\forall z, \forall R  : w'_{z,R} \geq 0$.
Therefore  $\{w'_{z,R} : z\in \mcZ, R\in \mcR\}$ is feasible for the primal of $\prt_\epsilon(f)$. Hence,
$$\prt_\epsilon(f) \leq \sum_{z} \sum_{R }  w'_{z,R} = \sum_r q(r) \cdot |\mcR_r| \leq 2^c \sum_r q(r) = 2^c \enspace .$$

\item  Fix $z' \in \mcZ$. We will show that $\srec^{z'}_\epsilon(f) \leq \prt_\epsilon(f)$; this will imply $\srec_\epsilon(f) \leq \prt_\epsilon(f)$.  Let $\{w_{z,R}: z\in \mcZ, R \in \mcR\}$ be an optimal solution of the primal for $ \prt_\epsilon(f)$. Let us define $\forall R \in \mcR :  w_R \defeq w_{z',R}$, hence $\forall R \in \mcR, w_R \geq 0$. Now,
\begin{align*}
 \forall (x,y) \in f^{-1}(z') : \sum_{R: (x,y) \in R} w_{z',R} \geq 1 - \epsilon &\quad \Rightarrow \quad  \sum_{R: (x,y) \in R} w_{R} \geq 1 - \epsilon ,\\
 \forall (x,y) \in f^{-1} - f^{-1}(z') :  \sum_{R: (x,y) \in R} w_{f(x,y),R} \geq 1 - \epsilon  &\quad  \Rightarrow  \quad  \sum_{R: (x,y) \in R} w_{R} \leq \epsilon ,\\
\forall (x,y) : \sum_{R: (x,y) \in R} \quad \sum_{z }  w_{z,R} = 1 &\quad  \Rightarrow  \quad  \sum_{R: (x,y) \in R} w_{R} \leq 1 \enspace .
\end{align*}
Hence $\{w_R : R\in \mcR\}$ forms a feasible solution to the primal for $\srec^z_\epsilon(f)$ which implies
$$\srec^z_\epsilon(f) \leq \sum_{R} w_R \leq \sum_{z} \sum_{R }  w_{z,R}  = \prt_\epsilon(f) \enspace .$$

\item  Fix $z \in \mcZ$. Since the primal program for $\srec^z_\epsilon(f)$ has extra constraints over the primal program for $\rec^z_\epsilon(f)$, it implies that $\rec^z_\epsilon(f) \leq \srec^z_\epsilon(f)$. Hence $\rec_\epsilon(f) \leq \srec_\epsilon(f)$.

\item (Sketch) Let $W \defeq \{w_{z,R}\}$ be an optimal solution to the primal for $\prt_0{f}$. It is easily seen that
$$w_{z,R} > 0 \Rightarrow ((x,y)\in R \Rightarrow f(x,y)=z) \enspace .$$
 Using standard Chernoff type arguments we can argue that there exists subset $W' \subseteq W$ with $|W'| = O(n \prt_0{f})$ such that :
$$\forall (x,y) \in f^{-1}: \sum_{R: (x,y) \in R, w_{f(x,y),R} \in W'} w_{f(x,y),R} > 0 \enspace .$$
Hence $W'$ is a cover of $\mcX \times \mcY$ using monochromatic rectangles. Now using arguments as in Theorem 2.11 of~\cite{kushilevitz&nisan:cc} it follows that $\D(f) = O((\log \prt_0{f} + \log n)^2)$.


\item Define $\mu_{x,y} \defeq 1; \phi_{x,y} \defeq 0 $ iff $(x,y) \in S$ and $\mu_{x,y} = \phi_{x,y} \defeq 0 $ otherwise. Since no two elements of $S$ can appear in the same rectangle, it is easily seen that the constraints for the dual of $\prt_0(f)$ are satisfied by $\{\mu_{x,y}, \phi_{x,y}\}$ . Hence $\prt_0(f) \geq \sum_{(x,y)} (\mu_{x,y} - \phi_{x,y}) = |S|$.

\end{enumerate}

\end{proof}

The following lemma shows the equivalence of the two definitions of the rectangle bound.
\begin{Lem} \label{lem:recsim}
Let $f : \mcX \times \mcY \rightarrow \mcZ$ be a function and let $\epsilon > 0$.  Then for all $z\in \mcZ$,
\begin{enumerate}
\item $\rec^z_\epsilon(f) \leq \trec^z_{\frac{\epsilon}{2}}(f)$.
\item $ \rec^z_\epsilon(f)  \geq \frac{1}{2} \cdot (\frac{1}{2} - \epsilon) \cdot \trec^z_{2\epsilon}(f)$.
\end{enumerate}
\end{Lem}
\begin{proof}
\begin{enumerate}
\item Fix $z\in\mcZ$. Let $k \defeq \rec^z_\epsilon(f)$. Let $\{\mu_{x,y} : (x,y) \in \mcX\times\mcY\}$ be an optimal solution to the dual for $\rec^z_\epsilon(f)$. We can assume without loss of generality that $(x,y) \notin f^{-1} \Rightarrow \mu_{x,y} = 0$. Let $k_1 \defeq \sum_{(x,y) \in f^{-1}(z)} \mu_{x,y}$ and $k_2 \defeq \sum_{(x,y)\in f^{-1} - f^{-1}(z)} \mu_{x,y}$.  Then,
\begin{eqnarray}
& & k  =(1-\epsilon)\sum_{(x,y)\in f^{-1}(z)} \mu_{x,y}  - \epsilon \sum_{(x,y) \in f^{-1} - f^{-1}(z)} \mu_{x,y}  \nonumber \\
& \Rightarrow & k = (1-\epsilon) k_1  - \epsilon k_2 \nonumber \\
&\Rightarrow & k_1 \geq k   \mbox{ and } k_1  \geq  \epsilon k_2 \quad \mbox{(since $k, k_2 \geq 0$)}\enspace . \label{eq:1}
\end{eqnarray}
Let us define $\lambda_{x,y} \defeq \frac{\mu_{x,y}}{2k_1}$ iff $f(x,y) = z$ and $\lambda_{x,y} \defeq \frac{\mu_{x,y}}{2k_2}$, otherwise. It is easily seen that $\lambda$ is a distribution on $\mcX \times \mcY \cap f^{-1}$ and $\lambda(f^{-1}(z)) = 0.5$. For all $R \in \mcR$,
\begin{eqnarray*}
\sum_{(x,y)\in f^{-1}(z)\cap R} \mu_{x,y}  - \sum_{(x,y)\in (R\cap f^{-1}) - f^{-1}(z)} \mu_{x,y} & \leq & 1 \\
\Rightarrow \sum_{(x,y)\in f^{-1}(z)\cap R} 2k_1\lambda_{x,y}  - \sum_{(x,y)\in R - f^{-1}(z)} 2k_2\lambda_{x,y} & \leq & 1 \\
\Rightarrow \sum_{(x,y)\in f^{-1}(z)\cap R} 2k_1\lambda_{x,y}  - \sum_{(x,y)\in R - f^{-1}(z)}  \frac{2k_1}{\epsilon} \lambda_{x,y} & \leq & 1  \quad \mbox{(from (\ref{eq:1}))}\\
\Rightarrow {\epsilon}\left(\sum_{(x,y)\in f^{-1}(z)\cap R} \lambda_{x,y} - \frac{1}{2k_1}\right)& \leq & \sum_{(x,y)\in R - f^{-1}(z)} \lambda_{x,y} \\
\Rightarrow {\epsilon}\left(\sum_{(x,y)\in f^{-1}(z)\cap R} \lambda_{x,y} - \frac{1}{2k}\right)& \leq & \sum_{(x,y)\in R - f^{-1}(z)} \lambda_{x,y} \quad \mbox{(from (\ref{eq:1}))}
\end{eqnarray*}
Let $R\in\mcR$ be such that $\sum_{(x,y)\in f^{-1}(z)\cap R} \lambda_{x,y} \geq \frac{1}{k}$. Then we have from above
\begin{equation} \frac{\epsilon}{2}\left(\sum_{(x,y)\in f^{-1}(z)\cap R} \lambda_{x,y} \right)  \leq  \sum_{(x,y)\in R - f^{-1}(z)} \lambda_{x,y} \enspace .\label{eq:2} \end{equation}
Therefore by definition $\trec^{z,\lambda}_{\frac{\epsilon}{2}}(f) \geq k$ which implies $\trec^z_{\frac{\epsilon}{2}}(f) \geq k$.

\item Fix $z\in\mcZ$. Let $k = \trec^z_{2\epsilon}(f)$. Let $\lambda$ be a distribution on $\mcX\times\mcY \cap f^{-1}$ such that $\trec^z_{2\epsilon}(f) = \trec^{z,\lambda}_{2\epsilon}(f)$ and $\lambda(f^{-1}(z)) \geq 0.5$. Let us define $\mu_{x,y} \defeq k \cdot \lambda_{x,y}$ iff $f(x,y) =z$; $\mu_{x,y} \defeq k \cdot \frac{\lambda_{x,y}}{2\epsilon}$ iff $(x,y) \in f^{-1} - f^{-1}(z)$ and $\mu_{x,y} =0$ otherwise. Now let $R\in\mcR$ be such that $\lambda(f^{-1}(z)\cap R ) \leq \frac{1}{k}$, then
\begin{eqnarray*}
\sum_{(x,y) \in f^{-1}(z)\cap R}\lambda_{x,y}  \leq \frac{1}{k} \quad \Rightarrow \quad \sum_{(x,y) \in f^{-1}(z)\cap R}\mu_{x,y} \leq 1 \enspace .
\end{eqnarray*}
Let $\lambda(f^{-1}(z)\cap R ) > \frac{1}{k}$, then
\begin{eqnarray*}
2\epsilon\sum_{(x,y) \in f^{-1}(z)\cap R}\lambda_{x,y}  & \leq & \sum_{(x,y) \in R - f^{-1}(z)} \lambda_{x,y} \\
\Rightarrow  \sum_{(x,y) \in f^{-1}(z)\cap R}\mu_{x,y} & \leq & \sum_{(x,y) \in (R\cap f^{-1}) - f^{-1}(z)} \mu_{x,y} \enspace .
\end{eqnarray*}
Hence the constraints of the dual program for $\rec^z_\epsilon(f)$ are satisfied by $\{\mu_{x,y} : (x,y) \in \mcX\times\mcY\}$. Now,
\begin{eqnarray*}
\rec^z_\epsilon(f) & \geq & \sum_{(x,y) \in f^{-1}(z)}  (1-\epsilon) \cdot  \mu_{x,y}  - \sum_{(x,y) \in f^{-1} - f^{-1}(z)} \epsilon \cdot \mu_{x,y} \\
& = & k \cdot \left(\sum_{(x,y) \in f^{-1}(z)}  (1-\epsilon) \cdot  \lambda_{x,y}  - \sum_{ (x,y) \in f^{-1} - f^{-1}(z)} \frac{\lambda_{x,y}}{2} \right) \\
& \geq& \frac{k}{2} \cdot (\frac{1}{2} - \epsilon) \quad \mbox{(since $\lambda(f^{-1}(z)) \geq 0.5$)}\enspace .
\end{eqnarray*}
\end{enumerate}
\end{proof}

The following lemma shows the equivalence of the two definitions of the smooth-rectangle bound.
\begin{Lem} \label{lem:srecsim} Let $f : \mcX \times \mcY \rightarrow \mcZ$ be a function and let $\epsilon >0$.  Then for all $z\in \mcZ$,
\begin{enumerate}
\item $\srec^{z}_\epsilon(f) \leq \tsrec^z_{\frac{\epsilon}{2},\frac{1-\epsilon}{2}}(f)$.
\item $ \srec^z_\epsilon(f)  \geq \frac{1}{2} \cdot (\frac{1}{4} - \epsilon) \cdot \tsrec^z_{2\epsilon,\frac{\epsilon}{2}}(f)$.
\end{enumerate}
\end{Lem}
\begin{proof}
\begin{enumerate}
\item  Fix $z\in\mcZ$.   Let $\{\mu_{x,y}, \phi_{x,y} : (x,y) \in \mcX\times\mcY\}$ be an optimal solution to the dual for $\srec^z_\epsilon(f)$. We can assume w.l.o.g. that $(x,y)\notin f^{-1} \Rightarrow \mu_{x,y} = \phi_{x,y}=0$; also  that $(x,y) \notin f^{-1}(z) \Rightarrow \phi_{x,y}=0$. Let us observe that we can assume w.l.o.g.~that $\forall (x,y) \in f^{-1}(z)$, either $\mu_{x,y} = 0$ or $\phi_{x,y} = 0$.  Otherwise let us say that for some $(x,y) \in f^{-1}(z): \mu_{x,y} \geq \phi_{x,y} > 0$. Then using $\mu_{x,y}' \defeq \mu_{x,y} - \phi_{x,y}$ and $\phi_{x,y}' \defeq 0$ instead of $(\mu_{x,y}, \phi_{x,y})$, and the rest the same, is a strictly better solution; that is the objective function is strictly larger in the new case. A similar argument can be made if for some $(x,y) \in f^{-1}(z): \phi_{x,y} \geq \mu_{x,y} > 0$.

Let $g: \mcX \times \mcY \rightarrow \mcZ$ be such that $g(x,y) = f(x,y)$ iff $\phi_{x,y}=0$ and $g(x,y) \neq f(x,y)$ otherwise ($g$ remains undefined wherever $f$ is undefined). For all $(x,y)$ let $\mu'_{x,y} \defeq \mu_{x,y}$ iff $\phi_{x,y}=0$ and $\mu'_{x,y} = \phi_{x,y}$ otherwise. Then $\forall (x,y), \mu'_{x,y} \geq 0$ and
\begin{align}
& \forall R\in\mcR : \sum_{(x,y)\in f^{-1}(z)\cap R} (\mu_{x,y} - \phi_{x,y}) - \sum_{(x,y)\in (R\cap f^{-1}) - f^{-1}(z)} \mu_{x,y} \leq 1 \nonumber \\
& \Rightarrow \forall R\in\mcR : \sum_{(x,y)\in g^{-1}(z)\cap R} \mu'_{x,y}  - \sum_{(x,y)\in (R\cap g^{-1}) - g^{-1}(z)} \mu'_{x,y} \leq 1 \enspace . \label{eq:srec1}
\end{align}
Hence $\{\mu'_{x,y} : (x,y) \in \mcX\times\mcY\}$ is a feasible solution to the dual of $\rec_\epsilon^z(g)$. Now,
\begin{eqnarray}
k & \defeq & \sum_{(x,y) \in g^{-1}(z)}  (1-\epsilon) \cdot  \mu'_{x,y}  - \sum_{(x,y) \in g^{-1} - g^{-1}(z)} \epsilon \cdot \mu'_{x,y} \label{eq:srec2} \\
& = & \sum_{(x,y) \in f^{-1}(z)}  (1-\epsilon) \cdot  \mu_{x,y} - \sum_{(x,y) \in f^{-1}(z)}  \epsilon \cdot  \phi_{x,y}  - \sum_{(x,y) \in f^{-1} - f^{-1}(z)} \epsilon \cdot \mu_{x,y} \nonumber \\
& \geq &  \sum_{(x,y) \in f^{-1}(z)}  (1-\epsilon) \cdot  \mu_{x,y} - \sum_{(x,y) \in f^{-1}(z)}  \phi_{x,y}  - \sum_{(x,y) \in f^{-1} - f^{-1}(z)} \epsilon \cdot \mu_{x,y} \nonumber \\
& = & \srec_\epsilon^z(f) \enspace . \label{eq:srec3}
\end{eqnarray}
Let $k_1 \defeq \sum_{(x,y) \in g^{-1}(z)} \mu'_{x,y}$ and $k_2 \defeq \sum_{(x,y) \in g^{-1} - g^{-1}(z)} \mu'_{x,y}$. Let $\lambda_{x,y} \defeq \frac{\mu'_{x,y}}{2k_1}$ iff $g(x,y) = z$ and $\lambda_{x,y} \defeq \frac{\mu'_{x,y}}{2k_2}$, otherwise. It is clear that $\lambda$ is a distribution on $\mcX\times\mcY \cap g^{-1}$ and $\lambda(g^{-1}(z))=0.5$. As in the proof of Part 1. of Lemma~\ref{lem:recsim}, using (\ref{eq:srec1}) and (\ref{eq:srec2}), we can argue that $\trec^{z,\lambda}_{\frac{\epsilon}{2}}(g) \geq \rec^{z,\lambda}_{\epsilon}(g) \geq k$. Also since $\sum_{(x,y)\in f^{-1}} ((1-\epsilon)\mu_{x,y} - \phi_{x,y}) \geq 0$ and $\sum_{(x,y)\in f^{-1}(z)} (\mu_{x,y} - \phi_{x,y}) - \sum_{(x,y)\in (f^{-1}) - f^{-1}(z)} \mu_{x,y} \leq 1$ we can argue that $\sum_{(x,y)\in f^{-1}(z)} \phi_{x,y}  \leq (1-\epsilon)k_2$ (we assume $\srec^z_\epsilon(f)$ is at least a large constant) . Therefore,
\begin{align*}
\Pr_{(x,y) \leftarrow \lambda}[g(x,y)\neq f(x,y)] = \sum_{(x,y)\in f^{-1}(z)} \frac{\phi_{x,y}}{2k_2} \leq \frac{1-\epsilon}{2} \enspace .
\end{align*}
Hence by definition, $\tsrec^z_{\frac{\epsilon}{2},\frac{1-\epsilon}{2}}(f) \geq \tsrec^{z,\lambda}_{\frac{\epsilon}{2},\frac{1-\epsilon}{2}}(f) \geq \trec^{z,\lambda}_{\frac{\epsilon}{2}}(g) \geq k \geq \srec_\epsilon^z(f)$. The last inequality follows from (\ref{eq:srec3}).

\item  Fix $z\in\mcZ$. Let $k \defeq \tsrec^z_{2\epsilon,\frac{\epsilon}{2}}(f)$. Let $\lambda$ be distribution on $\mcX\times\mcY \cap f^{-1}$ such that $\tsrec^z_{2\epsilon,\frac{\epsilon}{2}}(f) = \tsrec^{z,\lambda}_{2\epsilon,\frac{\epsilon}{2}}(f)$. Let $g : \mcX \times \mcY \rightarrow \mcZ$ be a function such that $\tsrec^{z,\lambda}_{2\epsilon,\frac{\epsilon}{2}}(f) = \rec^{z,\lambda}_{2\epsilon}(g)$ and $\lambda(g^{-1}(z)) \geq 0.5 $ and $\lambda(f\neq g) \leq \epsilon/2$. Note that we can assume w.l.o.g.~that $g(x,y) \neq f(x,y) \Rightarrow f(x,y) = z$.

For $(x,y)\in f^{-1}$, let us define $\mu_{x,y} \defeq k \cdot \lambda_{x,y}$ iff $g(x,y) = f(x,y) =z$ and $\mu_{x,y} \defeq k \cdot \frac{\lambda_{x,y}}{2\epsilon}$ iff $f(x,y) \neq z$. Let $\phi_{x,y} \defeq k \cdot \frac{\lambda_{x,y}}{2\epsilon}$ iff $z = f(x,y) \neq g(x,y)$. For $(x,y)\notin f^{-1}$, let $\mu_{x,y}=\phi_{x,y}=0$. Now let $R\in\mcR$ be such that $\lambda(g^{-1}(z)\cap R ) \leq \frac{1}{k}$, then
\begin{eqnarray*}
\sum_{(x,y) \in g^{-1}(z)\cap R}\lambda_{x,y}  \leq \frac{1}{k} \quad \Rightarrow \quad \sum_{(x,y) \in g^{-1}(z)\cap R}\mu_{x,y} & \leq & 1 \\
\Rightarrow  \sum_{(x,y) \in f^{-1}(z)\cap R}\mu_{x,y} - \phi_{x,y} & \leq & 1 \enspace .
\end{eqnarray*}
Let $\lambda(g^{-1}(z)\cap R ) > \frac{1}{k}$, then
\begin{eqnarray*}
2\epsilon\sum_{(x,y) \in g^{-1}(z)\cap R}\lambda_{x,y}  & \leq & \sum_{(x,y) \in R - g^{-1}(z)} \lambda_{x,y} \\
\Rightarrow  \sum_{(x,y) \in g^{-1}(z)\cap R}\mu_{x,y} & \leq & \sum_{(x,y) \in R - g^{-1}(z)} \mu_{x,y} + \phi_{x,y} \\
\Rightarrow  \sum_{(x,y) \in f^{-1}(z)\cap R}\mu_{x,y}  - \phi_{x,y}& \leq & \sum_{(x,y) \in (R\cap f^{-1}) - f^{-1}(z)} \mu_{x,y} \enspace .
\end{eqnarray*}
Hence the constraints of the dual program for $\srec^z_\epsilon(f)$ are satisfied by $\{\mu_{x,y}, \phi_{x,y} : (x,y)\in\mcX\times\mcY\}$. Now,
\begin{eqnarray*}
\srec^z_\epsilon(f) & \geq & \sum_{(x,y) \in f^{-1}(z)}  \left((1-\epsilon) \cdot  \mu_{x,y}  -\phi_{x,y} \right) - \sum_{(x,y) \in f^{-1} - f^{-1}(z)} \epsilon \cdot \mu_{x,y} \\
& \geq & \sum_{(x,y) \in g^{-1}(z)}  (1-\epsilon) \cdot  \mu_{x,y}  - \sum_{(x,y) \in f^{-1}(z)} \phi_{x,y} - \sum_{(x,y) \notin g^{-1}(z)} \epsilon \cdot \mu_{x,y} \\
& = & k \cdot \left(\sum_{(x,y) \in g^{-1}(z)}  (1-\epsilon) \cdot  \lambda_{x,y}  -  \frac{1}{2\epsilon}\sum_{(x,y) : f(x,y) \neq g(x,y)} \lambda_{x,y}  - \sum_{(x,y) \notin g^{-1}(z)} \frac{\lambda_{x,y}}{2} \right) \\
& \geq& \frac{k}{2} \cdot (\frac{1}{4} - \epsilon) \enspace .
\end{eqnarray*}
The last inequality follows since $\lambda(g^{-1}(z)) \geq 0.5$  and $\lambda(f\neq g) \leq \epsilon/2$.
\end{enumerate}
\end{proof}


The following lemma shows the equivalence of the two definitions of smooth-discrepancy.
\begin{Lem}\label{lem:sdisceq}
Let $f : \mcX\times\mcY\rightarrow \{0,1\}$ be a function and let $\epsilon > 0$. Then
\begin{enumerate}
\item $\tsdisc_{\frac{1}{2} - \frac{\epsilon}{8} }(f) \geq \sdisc_\epsilon(f) $.
\item $\frac{1}{2} \cdot\tsdisc_{\frac{1}{4+2\epsilon} }(f) \leq \sdisc_\epsilon(f) $.
\end{enumerate}
\end{Lem}
\begin{proof}
\begin{enumerate}
\item Let $k \defeq \sdisc_\epsilon(f)$. Let $\{\mu_{x,y}, \phi_{x,y}\}$ be an optimal solution to the dual for $\sdisc_\epsilon(f)$.
As in the proof of Lemma~\ref{lem:srecsim}, we can argue that for all $(x,y) \in f^{-1}$, either $\mu_{x,y} = 0$ or $\phi_{x,y} = 0$.  For $(x,y) \in f^{-1}$, let us define $\lambda'_{x,y} \defeq \max\{\mu_{x,y},\phi_{x,y}\}$ and let $\lambda_{x,y} \defeq \frac{\lambda'_{x,y}}{\sum_{(x,y)\in f^{-1}} \lambda'_{x,y}}$. It is clear that $\lambda$ is a distribution on $f^{-1}$. Let us define $g : \mcX\times\mcY\rightarrow \{0,1\}$ such that $g^{-1} = f^{-1}$. For  $(x,y) \in f^{-1}$, let $g(x,y) = f(x,y) $ iff $\phi_{x,y}=0$ and let $g(x,y) \neq f(x,y) $ iff $\phi_{x,y} \neq 0$. Now
\begin{align*}
& \forall R : |\sum_{(x,y)\in f^{-1}(1)\cap R} (\mu_{x,y} - \phi_{x,y}) - \sum_{(x,y)\in R \cap f^{-1}(0) } (\mu_{x,y} - \phi_{x,y}) | \leq 1 \\
 \Rightarrow \quad & \forall R : |\sum_{(x,y)\in g^{-1}(1)\cap R} \lambda'_{x,y} - \sum_{(x,y)\in R \cap g^{-1}(0) } \lambda'_{x,y} | \leq 1 \\
 \Rightarrow \quad & \forall R : |\sum_{(x,y)\in g^{-1}(1)\cap R} \lambda_{x,y} - \sum_{(x,y)\in R \cap g^{-1}(0) } \lambda_{x,y} | \leq \frac{1}{\sum_{x,y} \mu_{x,y} + \phi_{x,y}} \leq \frac{1}{k} \enspace .
\end{align*}
Hence $\disc^\lambda(g) \geq k$. Also since $\sum_{(x,y)} \mu_{x,y} - (1+\epsilon) \phi_{x,y} \geq 0$,
\begin{align*}
\Pr_{(x,y) \leftarrow \lambda} [g(x,y) \neq f(x,y)] = \frac{1}{\sum_{x,y} \mu_{x,y} + \phi_{x,y}} \sum_{(x,y)} \phi_{x,y} < \frac{1}{2+\epsilon} \leq \frac{1}{2} - \frac{\epsilon}{8} \enspace .
\end{align*}
Hence our result.

\item Let $\delta \defeq \frac{1}{4+2\epsilon}$. Let $\lambda$ be a distribution on $f^{-1}$ such that  $k \defeq \tsdisc_{\delta}(f) = \tsdisc^\lambda_{\delta}(f) $ and $\Pr_{(x,y) \leftarrow \lambda} [g(x,y) \neq f(x,y)] < \delta$. For $(x,y) \in f^{-1}$, let $\mu_{x,y} \defeq k \cdot \lambda_{x,y}; \phi_{x,y} =0$ iff $f(x,y) = g(x,y)$ and $\phi_{x,y} \defeq k \cdot \lambda_{x,y}; \mu_{x,y} =0$ iff $f(x,y) \neq g(x,y)$. Then,
\begin{align*}
& \forall R : |\sum_{(x,y)\in g^{-1}(1)\cap R} \lambda_{x,y} - \sum_{(x,y)\in R \cap g^{-1}(0) } \lambda_{x,y} | \leq \frac{1}{k} \\
 \Rightarrow \quad & \forall R : |\sum_{(x,y)\in f^{-1}(1)\cap R} (\mu_{x,y} - \phi_{x,y}) - \sum_{(x,y)\in R \cap f^{-1}(0) } (\mu_{x,y} - \phi_{x,y}) | \leq 1 \enspace .
\end{align*}
Hence $\{\mu_{x,y}, \phi_{x,y}\}$ form a feasible solution to the dual for $\sdisc_\epsilon(f)$. Now,
\begin{align*}
\sdisc_\epsilon(f) \geq \sum_{(x,y)} \mu_{x,y} - (1+\epsilon) \phi_{x,y} > k((1-\delta) - (1+\epsilon) \delta ) = k(1 - (2+\epsilon) \delta ) = \frac{k}{2} \enspace .
\end{align*}
\end{enumerate}
\end{proof}

The following lemma states the rectangle bound dominates the discrepancy bound for Boolean functions and hence the smooth-rectangle bound dominates the smooth-discrepancy bound.
\begin{Lem}\label{lem:recgeqdisc}
Let $f : \mcX\times\mcY\rightarrow \{0,1\}$ be a function; let $z\in \{0,1\}$ and let $\lambda$ be a distribution on $\mcX\times\mcY \cap f^{-1}$. Let $\epsilon, \delta > 0$, then
$$\rec^{z}_\epsilon(f) \geq (\frac{1}{2} - \epsilon)\disc^\lambda(f)  - \frac{1}{2} \enspace .$$
This implies by definition and Lemma~\ref{lem:recsim},
$$\trec^{z}_{\frac{\epsilon}{2}}(f) \geq \rec^{z}_\epsilon(f) \geq (\frac{1}{2} - \epsilon)\disc(f)  - \frac{1}{2} \enspace ,$$
$$\Rightarrow \tsrec^{z}_{\frac{\epsilon}{2}, \delta}(f) \geq (\frac{1}{2} - \epsilon)\tsdisc_\delta(f)  - \frac{1}{2} \enspace .$$
\end{Lem}
\begin{proof}
Let $k \defeq \disc^\lambda(f)$. Let $\forall (x,y)\in f^{-1}: \mu_{x,y} \defeq k \cdot \lambda_{x,y}$ and $\mu_{x,y}=0$ otherwise. Then we have:
\begin{align*}
 \forall R : \sum_{(x,y) \in R \cap f^{-1}(z)} \lambda_{x,y} - \sum_{(x,y) \in R - f^{-1}(z)} \lambda_{x,y} & \leq \frac{1}{k} \\
 \Rightarrow \forall R : \sum_{(x,y) \in R \cap f^{-1}(z)} \mu_{x,y} - \sum_{(x,y) \in (R\cap f^{-1}) - f^{-1}(z)} \mu_{x,y} & \leq 1 \enspace .
\end{align*}
Hence the constraints for the dual of the linear program for $\rec^{z}_\epsilon(f)$ are satisfied by $\{ \mu_{x,y} : (x,y) \in \mcX\times\mcY\}$. Now,
\begin{align*}
\rec^{z}_\epsilon(f) & \geq \sum_{(x,y) \in f^{-1}(z)}  (1-\epsilon) \cdot  \mu_{x,y}  - \sum_{(x,y) \in f^{-1} - f^{-1}(z)} \epsilon \cdot \mu_{x,y} \\
& = k \cdot \left(\sum_{(x,y) \in f^{-1}(z)}  (1-\epsilon) \cdot  \lambda_{x,y}  - \sum_{(x,y) \notin f^{-1}(z)} \epsilon \cdot \lambda_{x,y} \right)\\
& = k \cdot \left(\sum_{(x,y) \in f^{-1}(z)}    \lambda_{x,y}  -  \epsilon \right) \\
& \geq  k \cdot \left(\frac{1}{2} - \frac{1}{2k}  -  \epsilon \right) = (\frac{1}{2} - \epsilon) k - \frac{1}{2} \enspace .
\end{align*}
The last inequality follows since $\disc^\lambda(f) = k$.
\end{proof}

For a function $g: \mcX\times\mcY \rightarrow \{0,1\}$, let $A_g$ be the sign matrix corresponding to $g$, that is $A_g(x,y) \defeq (-1)^{g(x,y)}$. Similarly for a sign matrix $A$, let $g_A$ be the corresponding function given by $g_A(x,y) \defeq (1 - A(x,y))/2$. For distribution $\lambda$ on $\mcX\times\mcY$, let $P_\lambda$ be the matrix defined by $P_\lambda(x,y) \defeq \lambda(x,y)$. For matrix $B$, define $\signorm{B} \defeq \sum_{i,j}|B(i,j)|$. For matrices $C,D$, let $C \circ D$ denote the entry wise {\em Hadamard product} of $C,D$.
Following lemma states the equivalence between smooth-discrepancy and the $\gamma_2$ bound.
\begin{Lem}\label{lem:sdisceqgamma2}
Let $f:\mcX\times\mcY \rightarrow \{0,1\}$ be a Boolean function and let $\alpha >1$. Then
$$\frac{1}{2}\cdot \tsdisc_{\frac{1}{2(\alpha+1) }}(f) \leq \gamma_2^\alpha(A_f) \leq 8 \cdot \tsdisc_{ \frac{1}{\alpha+1}}(f) \enspace .$$
\end{Lem}
\begin{proof}
 We have the following facts:
\begin{Fact}[\cite{linial:norms}]
For every sign matrix $A$,
$$\gamma_2^\alpha(A)  = \max_B \frac{1}{2\gamma_2^*(B)} \left((\alpha + 1)\ip{A}{B} - (\alpha-1)\signorm{B} \right) \enspace .$$
Above, $\gamma_2^*(\cdot)$ is the dual norm of $\gamma_2(\cdot)$.
\end{Fact}
\begin{Fact}[\cite{linial:norms}]
Let $A$ be a sign matrix and let $\lambda$ be a distribution. Then,
$$ \frac{1}{8\gamma_2^*(A\circ P_\lambda)} \leq \disc^\lambda(g_A) \leq \frac{1}{\gamma_2^*(A\circ P_\lambda)} \enspace . $$
\end{Fact}
Therefore we have,
\begin{align*}
\gamma_2^\alpha(A_f) & = \max_B \frac{1}{2\gamma_2^*(B)} \left((\alpha+1)\ip{A_f}{B} - (\alpha-1)\signorm{B} \right) \\
 & = \max_{B : \signorm{B} =1} \frac{1}{2\gamma_2^*(B)} \left((\alpha+1)\ip{A_f}{B} - (\alpha-1) \right) \\
   & = \max_{g,\lambda} \frac{1}{\gamma_2^*(A_g \circ P_\lambda)} \left(1 - (\alpha+1)\lambda(f \neq g) \right) \\
   & \leq \max_{g,\lambda} 8 \cdot \disc^\lambda(g) \left(1 - (\alpha+1)\lambda(f \neq g) \right) \\
    & \leq \max \{ 8 \cdot \disc^\lambda(g) : g, \lambda \mbox{ such that } \lambda(f \neq g) <  \frac{1}{\alpha+1} \} \\
    & = 8  \cdot \tsdisc_{\frac{1}{\alpha+1}}(f) \enspace .
\end{align*}
Similarly,
\begin{align*}
\gamma_2^\alpha(A_f) & = \max_{g,\lambda} \frac{1}{\gamma_2^*(A_g \circ P_\lambda)} \left(1 - (\alpha+1)\lambda(f \neq g) \right)  \\
   & \geq \max_{g,\lambda}  \disc^\lambda(g) \left(1 - (\alpha+1)\lambda(f \neq g) \right)  \\
    & \geq \max \{ \frac{1}{2} \cdot \disc^\lambda(g) : g, \lambda \mbox{ such that } \lambda(f \neq g) < \frac{1}{2(\alpha+1) } \} \\
    & =  \frac{1}{2}\cdot \tsdisc_{\frac{1}{2(\alpha+1) }}(f) \enspace .
\end{align*}
\end{proof}

From Lemma~\ref{lem:recgeqdisc} and Lemma~\ref{lem:sdisceqgamma2} we have the following corollary.
\begin{Cor}
Let $f:\mcX\times\mcY \rightarrow \{0,1\}$ be a Boolean function; let $z\in\{0,1\}$; let $\alpha >1, \epsilon>0$. Then,
$$ \tsrec^{z}_{\frac{\epsilon}{2}, \frac{1}{\alpha+1 }}(f) \geq (\frac{1}{2} - \epsilon)\frac{1}{8}\gamma_2^\alpha(A_f)  - \frac{1}{2}  \enspace .$$
\end{Cor}

\subsection{Partition bound for relations}
Here we define the partition bound for relations.
\begin{Def}[Partition Bound for relation]
Let $f \subseteq \mcX \times \mcY \times \mcZ$ be a relation. The $\epsilon$-partition bound of $f$, denoted $\prt_\epsilon(f)$, is given by the optimal value of the following linear program.

\vspace{0.2in}

{\footnotesize
\begin{minipage}{3in}
    \centerline{\underline{Primal}}\vspace{-0.1in}
    \begin{align*}
      \text{min:}\quad & \sum_{z} \sum_{R}  w_{z,R} \\
       \quad & \forall (x,y) :  \sum_{R: (x,y) \in R} \quad \sum_{z:(x,y,z) \in f} w_{z,R} \geq 1 - \epsilon,\\
      & \forall (x,y) : \sum_{R: (x,y) \in R} \quad \sum_{z}  w_{z,R} = 1 , \\
      & \forall z , \forall R  : w_{z,R} \geq 0 \enspace .
    \end{align*}
\end{minipage}
\begin{minipage}{3in}\vspace{-0.4in}
    \centerline{\underline{Dual}}\vspace{-0.1in}
    \begin{align*}
      \text{max:}\quad &  \sum_{(x,y)} (1-\epsilon) \mu_{x,y} + \phi_{x,y}\\
       \quad &  \forall z, \forall R : \sum_{(x,y) : (x,y)\in R; (x,y,z) \in f} \mu_{x,y}   + \sum_{(x,y)\in  R} \phi_{x,y}  \leq 1,\\
      & \forall (x,y) : \mu_{x,y} \geq 0, \phi_{x,y} \in \mbR \enspace .
    \end{align*}
\end{minipage}
}
\end{Def}

As in Theorem~\ref{thm:main}, we can show that partition bound is a lower bound on the communication complexity. Its proof is skipped since it is very similar.
\begin{Lem} Let $f \subseteq \mcX \times \mcY \times \mcZ$ be a relation. Then,
$\R^\pub_\epsilon(f) \geq \log \prt_\epsilon(f) \enspace .$
\end{Lem}

\subsection{Las Vegas Partition Bound}
In this section we consider the Las Vegas communication complexity. Las Vegas protocols use randomness and for each input they are allowed to output "don't know" with probability $1/2$, however when they do give an answer then it is required to be correct. An equivalent way to view is that these protocols are never allowed to err, but for each input we only count the expected communication (over the coins), instead of the worst case communication (as in deterministic protocols). Below we present a lower bound for Las Vegas protocols via a linear program, which we call the {\em Las Vegas partition bound}. 
\begin{Def}[Las Vegas Partition Bound]
Let $f : \mcX \times \mcY \rightarrow \mcZ$ be a partial function. The Las Vegas-partition bound of $f$, denoted $\prt_{LV}(f)$, is given by the optimal value of the following linear program. Let $\mcR_f$ denote the set of monochromatic rectangles for $f$.

\vspace{0.2in}

{\footnotesize
\begin{minipage}{3in}
    \centerline{ \underline{Primal}}\vspace{-0.2in}
 \begin{align*}
      \text{min:}\quad & \sum_{R \in \mcR_f}  w_{R} + \sum_{R \in \mcR} v_R\\
      \quad & \forall (x,y) \in f^{-1}: \sum_{R \in \mcR_f: (x,y) \in R} w_{R} \geq \frac{1}{2},\\
      & \forall (x,y) : \sum_{R \in \mcR_f: (x,y) \in R} \quad   w_{R}+ \sum_{R: (x,y) \in R}v_R = 1 , \\
      & \forall R  : w_{R},v_R \geq 0 \enspace .\\
    \end{align*}
\end{minipage}
\begin{minipage}{3in}\vspace{-0.2in}
    \centerline{\underline{Dual}}\vspace{-0.2in}
 \begin{align*}
      \text{max:}\quad &  \sum_{(x,y)\in f^{-1}} \frac{1}{2} \cdot \mu_{x,y} + \sum_{(x,y)}\phi_{x,y}\\
      \quad &  \forall R\in \mcR_f : \sum_{(x,y)\in f^{-1}(z)\cap R} \mu_{x,y}   + \sum_{(x,y)\in  R} \phi_{x,y}  \leq 1,\\
            \quad &  \forall R \in \mcR : \sum_{(x,y)\in  R} \phi_{x,y}  \leq 1,\\
      & \forall (x,y) : \mu_{x,y} \geq 0, \phi_{x,y} \in \mbR \enspace .
 \end{align*}
\end{minipage}
}
\end{Def}


The following lemma follows easily using arguments as before. Below $\R_0(f)$ represents the Las Vegas communication complexity of $f$; please refer to~\cite{kushilevitz&nisan:cc} for explicit definition of $\R_0(f)$.
\begin{Lem}Let $f : \mcX \times \mcY \rightarrow \mcZ$ be a partial function. Then, $\R_0(f) \geq \log\prt_{LV}$ . 
\end{Lem}
Let $\prt^*_{LV}(f)$ be defined similarly to $\prt_{LV}(f)$, except that the constraints $$\forall (x,y) \in f^{-1}: \sum_{R \in \mcR_f: (x,y) \in R} w_{R} \geq 1/2$$ are replaced by $$ \forall (x,y) \in f^{-1}: \sum_{R \in \mcR_f: (x,y) \in R} w_{R} = 1/2 \enspace .$$  Then we can observe $\prt_0(f) \geq \prt^*_{LV}(f) \geq \frac{1}{2}\prt_0(f)$. Note that $\log \prt^*_{LV}(f)$ forms a lower bound for $\R_0(f)$ if there is a Las Vegas protocol for $f$ that has the probability of output 'don't know' for all inputs. 

\subsection{Separations between bounds}
In this section, we discuss some separations between some of the bounds we mentioned.
\begin{Thm}
\begin{enumerate}
\item $\log \prt_\epsilon(\disj)\geq\log\rec_\epsilon(\disj)=\Omega(n)$, while $\log \gamma_2^\alpha(\disj)=O(\sqrt n)$ for all $\epsilon<1/2$ and $\alpha>1$.
\item There is a function $f:\{0,1\}^n\times \{0,1\}^n\to\{0,1\}$ such that $\log \prt_\epsilon(f)\geq\log\rec_\epsilon(f)=\Omega(n)$, while $\log\rank(f)=O(n^{0.62})$ for all $\epsilon<1/2$.
\item Let the function $\LNE:\{0,1\}^{n^2}\times \{0,1\}^{n^2}\to\{0,1\}$ be defined as 
$$\LNE(x_1,\ldots,x_n;y_1,\ldots,y_n)=1\iff \forall i:x_i\neq y_i$$, where all $x_i,y_j$ are strings of length $n$. Then $\D(\LNE)= \rank(\LNE)= n^2$, however $\R_0(\LNE) = O(n)$ and $\log\prt_0(\LNE)  =O(n)$. 
\end{enumerate} \label{thm:sep}
\end{Thm}
\begin{proof}
\begin{enumerate}
\item The lower bound is from \cite{razborov:disj}, the upper bound follows from \cite{aaronson:search}.
\item The function is described in \cite{nisan&wigderson:rank}.
\item 
The lower bound $\D(\LNE)=n^2$ is shown in~\cite{kushilevitz&nisan:cc} where it was shown that $\log \rank(\LNE) = n^2$. It is not hard to see that the Las Vegas complexity of $\LNE$ is $O(n)$ which is also shown in~\cite{kushilevitz&nisan:cc}. 

In order to show $\log\prt_0(\LNE)  =O(n)$, we describe a solution to the primal program for the partition bound for $\LNE$. We will assign a positive weight $w_R$, to every monochromatic rectangle $R$ such that the sum of weights is small. In this case one can set $w_{z,R} \defeq w_R$ where $z$ is the color of the monochromatic rectangle $R$ (all other $w_{z',R}$ are $0$). 

We present the analysis below assuming that none of $x_1\ldots x_n, y_1\ldots y_n$ is $0^n$. The analysis can be extended easily if such is the case. 

First we consider the 1-inputs of $\LNE$. Let $R_{z_1,\ldots,z_n,s_1,\ldots,s_n}$ be the rectangle that contains all inputs with $\sum_j x_i(j)\cdot z_i(j)=s_i\mod 2$ and  $\sum_j y_i(j)\cdot z_i(j)\neq s_i\mod 2$ for all $i$. Note that these are 1-chromatic rectangles. We give weight $2^n/2^{n^2}$ to each such rectangle. For every 1-input $x_1,\ldots,x_n;y_1,\ldots,y_n$ and all $s_1,\ldots, s_n$

\[\Pr_{z_1\ldots z_n}\left(\sum_j x_i(j)\cdot z_i(j)=s_i\mod 2\wedge\sum_j y_i(j)\cdot z_i(j) \neq s_i\mod 2\mbox{ for all }   i\right)=1/4^n\] for uniform $z_1,\ldots,z_n$.
Hence \[\sum w_{R_{z_1,\ldots,z_n,s_1,\ldots,s_n}}=2^n\cdot\frac{2^{n^2}}{4^n}\cdot\frac{2^n}{2^{n^2}}=1,\] when the sum is over all $R_{z_1,\ldots,z_n,s_1,\dots,s_n}$ consistent  with   $x_1,\ldots,x_n;y_1,\ldots,y_n$ . The sum of the weights $w_{R_{z_1,\ldots,z_n,s_1,\ldots,s_n}}$ of all such rectangles is exactly $2^{2n}$.

Now we turn to the 0-inputs. For each of them there is a position $k+1$, where $x_{k+1}=y_{k+1}$ but $x_i\neq y_i$ for all $i\leq k$. Let $R_{z_1,\ldots, z_k,s_1,\ldots, s_k,u}$ denote the rectangle that contains all inputs with $\sum_j x_i(j)\cdot z_i(j)=s_i\mod 2$ and  $\sum_j y_i(j)\cdot z_i(j)\neq s_i\mod 2$ for all $i\leq k$ and $x_{k+1}=y_{k+1}=u$. The rectangle $R_{z_1,\ldots, z_k,s_1,\ldots, s_k,u}$ receives weight $2^{k}/2^{nk}$. As before it can be argued that every 0-input lies in $2^{nk}/2^k$ such rectangles, so the constraints are satisfied. The overall sum of rectangle weights is at most
\[\sum_{k=0}^{n-1}2^{kn}\cdot{2^k}\cdot2^n\cdot \frac{2^k}{2^{kn}}\leq 2\cdot 2^{3n}.\]
Hence $\log \prt_0(\LNE) \leq \log \sum_{R \in \mcR_{\LNE}} w_R = O(n)$.
\end{enumerate}
\end{proof}

\section{Query Complexity Bounds}
In this section we define the partition bound for query complexity and also other previously known bounds.
\subsection{Definitions}
Let $f : \{0,1\}^n \rightarrow \{0,1\}^m$ be a function. Henceforth all functions considered are partial unless otherwise specified. An {\em assignment} $A : S \rightarrow \{0,1\}^m$ is an assignment of values to some subset $S$ of $n$ variables. We say that $A$ is {\em consistent} with $x\in\{0,1\}^n$ if $x_i= A(i)$ for all $i\in S$. We write $x\in A$ as shorthand for '$A$ is consistent with $x$'. We write $|A|$ to represent the size of $A$ which is the cardinality of $S$ (not to be confused with the number of consistent inputs). Furthermore we say that an index $i$ {\em appears} in $A$, iff $i \in S$ where $S$ is the subset of $[n]$ corresponding to $A$.
Let $\mcA$ denote the set of all assignments.  Below we assume $x \in \{0,1\}^n $, $A \in \mcA$ and $z \in \{0,1\}^m$, unless otherwise specified.

\subsubsection{Partition Bound}
\begin{Def}[Partition Bound]
Let $f : \{0,1\}^n \rightarrow \{0,1\}^m$ be a function and let $\epsilon \geq 0$. The $\epsilon$-partition bound of $f$, denoted $\prt_\epsilon(f)$, is given by the optimal value of the following linear program.

\vspace{0.2in}

{\footnotesize
\begin{minipage}{3in}
    \centerline{\underline{Primal}}\vspace{-0.2in}
    \begin{align*}
      \text{min:}\quad & \sum_{z} \sum_{A}  w_{z,A} \cdot 2^{|A|} \\
       \quad &  \forall x \in f^{-1} : \sum_{A: x \in A} w_{f(x),A} \geq 1 - \epsilon,\\
      & \forall x : \sum_{A: x \in A} \quad \sum_{z}  w_{z,A} = 1 , \\
      & \forall z, \forall A : w_{z,A} \geq 0 \enspace .
    \end{align*}
\end{minipage}
\begin{minipage}{3in}\vspace{-0.3in}
    \centerline{\underline{Dual}}\vspace{-0.2in}
        \begin{align*}
      \text{max:}\quad &  \sum_{x \in f^{-1}} (1-\epsilon) \mu_{x} + \sum_{x} \phi_{x}\\
       \quad &   \forall A, \forall z  : \sum_{x\in f^{-1}(z)\cap A} \mu_{x}   + \sum_{x \in  A} \phi_{x}  \leq 2^{|A|},\\
      & \forall x : \mu_{x} \geq 0, \phi_{x} \in \mbR \enspace .
    \end{align*}
\end{minipage}
}
\end{Def}

\subsubsection{Known Bounds}

In this section we define some known complexity measures of functions. All of these except the (errorless) certificate complexity are lower bounds for randomized query complexity. See the survey by Buhrman and de Wolf \cite{buhrman:dectreesurvey} for further information.

\begin{Def}[Certificate Complexity]
For $z\in\{0,1\}^m$, a $z$-certificate for $f$ is an assignment $A$ such that $x\in A \Rightarrow f(x)=z$.
The certificate complexity $\C_x(f)$ of $f$ on $x$ is the size of the smallest $f(x)$-certificate that is consistent with $x$. The certificate complexity of $f$ is $\C(f) \defeq \max_{x\in f^{-1}} \C_x(f)$. The $z$-certificate complexity of $f$ is  $\C^z(f) \defeq \max_{x : f(x)=z} \C_x(f)$.
\end{Def}

\begin{Def}[Sensitivity and Block Sensitivity]
For $x\in\{0,1\}^n$ and  $S\subseteq [n]$, let $x^S$ be $x$ flipped on locations in $S$. The sensitivity $\s_x(f)$ of $f$ on $x$ is the number of different $i\in[n]$ for which $f(x)\neq f(x^{\{i\}})$. The sensitivity of $f$ is $\s(f)\defeq \max_{x\in f^{-1}} \s_x(f)$.

The block sensitivity $\bs_x(f)$ of $f$ on $x$ is the maximum number $b$ such that there are disjoint sets $B_1, \ldots, B_b$ for which $f(x) \neq f(x^{B_i})$. The block sensitivity  of $f$ is $\bs(f) \defeq \max_{x\in f^{-1}} \bs_x(f)$. If $f$ is constant, we define $\s(f) = \bs(f) =0$.
It is clear from definitions that $\s(f) \leq \bs(f)$.
\end{Def}

\begin{Def}[Randomized Certificate Complexity~\cite{aaronson:cert}]
A $\epsilon$-error randomized verifier for  $x\in\{0,1\}^n$ is a randomized
algorithm that, on input $y\in\{0,1\}^n$, queries $y$ and (i) accepts
with probability $1$ if $y=x$, and (ii) rejects with probability
at least $1 - \epsilon$ if $f(y) \neq f(x)$. If $y \neq x$ but
$f(y) = f(x)$, the acceptance probability can be arbitrary.
Then $\RC^x_\epsilon(f)$ is the maximum number of
queries used by the best $\epsilon$-error randomized verifier for $x$, and $\RC_\epsilon(f) \defeq \max_{x\in f^{-1}} \RC^x_\epsilon(f)$.
\end{Def}
The above definition is stronger than the one in~\cite{aaronson:cert}.
\begin{Def}[Approximate Degree]
Let $f : \{0,1\}^n \rightarrow \{0,1\}$ be a Boolean function and let $\epsilon  > 0$. A polynomial $\mbR^n \rightarrow \mbR$ is said to $\epsilon$-approximate $f$, if
$|p(x) -f(x)| < \epsilon$ for all $x \in f^{-1}$ and $0\leq p(x) \leq 1$ for all $x\in \{0,1\}^n$. The $\epsilon$-approximate degree $\tdeg_\epsilon(f)$ of $f$ is the minimum degree among all multi linear polynomials that $\epsilon$-approximate $f$. If $\epsilon=0$ we write $\edeg(f)$.
\end{Def}

\begin{Def}[Classical Adversary Bound]
Let  $f : \{0,1\}^n \rightarrow \{0,1\}^m$ be a function. Let $p = \{p_x: x \in \{0,1\}^n, p_x \mbox{ is a probability distribution on } [n] \}$. The classical adversary bound for $f$ denoted $\cadv(f)$, is defined as
\[\cadv(f) \defeq \min_{p} \max_{x,y:f(x)\neq f(y)}\frac{1}{\sum_{i:x_i\neq y_i}\min\{p_x(i),p_y(i)\} } \enspace .\]
\end{Def}

The classical adversary bound is defined in an equivalent but slightly different way by Laplante and Magniez \cite{laplante:kolmogorov}; the above formulation appears in their proof and is made explicit in \cite{spalek:equivalent}. Aaronson \cite{aaronson:local} defines a slightly weaker version as observed in \cite{laplante:kolmogorov}. Laplante and Magniez do not show an general upper bound for the classical adversary bound, but it is easy to see that $\cadv(f)=O(\C(f))$ for all total functions.

\begin{Def}[Quantum Adversary Bound]
Let  $f : \{0,1\}^n \rightarrow \mcZ$ be a function. Let $\Gamma$ be a Hermitian matrix whose rows and columns are labeled by elements in $\{0,1\}^n$, such that $\Gamma(x,y) = 0$ whenever $f(x) \neq f(y)$. For $i \in [n]$, let $D_i$ be a Boolean matrix whose rows and columns are labeled by elements in $\{0,1\}^n$, such that $D_i(x,y) = 1$ if $x_i \neq y_i$ and
$D_i(x,y) = 0$ otherwise. The {\em quantum adversary bound} for $f$, denoted $\adv(f)$ is defined as
$$ \adv \defeq \max_{\Gamma \neq 0} \frac{\norm{\Gamma}}{\max_i \norm{\Gamma \circ D_i}} \enspace .$$
\end{Def}

\suppress{
\begin{Def}[Classical Adversary Bound] Let $f : \{0,1\}^n \rightarrow \{0,1\}$ be a Boolean function. Let $\lambda$ be a distribution on $f^{-1}(0) \times f^{-1}(1)$. For $x \in f^{-1}(0)$ and $i \in [n]$, let
$$ \theta_\lambda(x,i) \defeq  \frac{\sum_{y: y_i\neq x_i} \lambda_{x,y}}{ \sum_{y}\lambda_{x,y}}  = \Pr_{y \leftarrow \lambda_x}[y_i \neq x_i ]\enspace .$$
For $y \in f^{-1}(1)$ and $i \in [n]$, let
$$ \theta_\lambda(y,i) \defeq  \frac{\sum_{x: x_i\neq y_i} \lambda_{x,y}}{ \sum_{x}\lambda_{x,y}}  = \Pr_{x \leftarrow \lambda_y}[y_i \neq x_i ]\enspace .$$
The classical adversary bound of $f$, denoted $\cadv(f)$ is defined as:
$$ \cadv(f) \defeq \min_{\lambda} \quad \max_{x\in f^{-1}(0), y\in f^{-1}(1), i \in [n]: x_i\neq y_i} \quad \min\{\theta_\lambda(x,i), \theta_\lambda(y,i)\} \enspace .$$
\end{Def}
}

\subsection{Comparison between bounds}

The following theorem captures the key relations between the above bounds. Below $\R_\epsilon(f)$ denotes the $\epsilon$-error randomized query complexity of $f$.
\begin{Thm}Let $f : \{0,1\}^n \rightarrow \{0,1\}^m$ be a function, then
\begin{enumerate}
\item  $\R_\epsilon(f) \geq \frac{1}{2} \log \prt_\epsilon(f)$.
\item  $\log \prt_0(f) \geq \C(f) $.
\item Let $\epsilon < 1/2$, then $ \log \prt_{\frac{\epsilon}{4}}(f)  \geq   \epsilon \cdot \bs(f) + \log \epsilon -2$.
\item $ \log \prt_{\epsilon}(f)  \geq   \RC_{\frac{2\epsilon}{1-2\epsilon}}(f) + \log \epsilon $.
\item $\log \prt_{\epsilon}(f) \geq (1-4\epsilon)\cdot \cadv(f) + \log \epsilon$.
\item Let $f : \{0,1\}^n \rightarrow \{0,1\}$ be a  Boolean function. Then, $\log \prt_\epsilon(f) \geq  \tdeg_{2\epsilon}(f) + \log \epsilon$.
\item Let $f : \{0,1\}^n \rightarrow \{0,1\}$ be a  Boolean function. Then, $\D(f) = O(\log \prt_0(f) \cdot \log \prt_{1/3}(f))$ and  $\D(f) = O(\log \prt_{1/3}(f)^3)$, where $\D(f)$ represents the deterministic query complexity of $f$.
\end{enumerate} \label{thm:mainquery}
\end{Thm}

\begin{proof}
\begin{enumerate}
\item Let $\{w_{z,A}\}$ be an optimal solution to the primal of $\prt_\epsilon(f)$. Let $\mcP$ be a randomized algorithm which achieves $\R_\epsilon(f)$. Then $\mcP$ is a convex combination of deterministic algorithms where each deterministic algorithm is a decision tree of depth at most $\R_\epsilon(f)$. As in the proof of Part 1. of Theorem~\ref{thm:main}, we can argue that $\sum_{z} \sum_{A}  w_{z,A}  \leq  2^{\R_\epsilon(f)}$.   Now since for each $A$ above $|A|\leq \R_\epsilon(f)$,
\begin{align*}
\prt_\epsilon(f) = \sum_{z} \sum_{A}  w_{z,A} 2^{|A|}  \leq 2^{\R_\epsilon(f)} \left(\sum_{z\in \{0,1\}^m} \sum_{A\in \mcA}  w_{z,A} \right)   \leq 2^{2\R_\epsilon(f)} \enspace .
\end{align*}
Hence our result.
\item Let $\{w_{z,A}\}$ be an optimal solution to the primal of $\prt_0(f)$. It is easily observed that $w_{z,A} > 0$ implies that $A$ is a $z$-certificate. Fix $x\inf^{-1}$, now
\begin{align*}
\prt_0(f) = \sum_{z} \sum_{A}  w_{z,A} \cdot 2^{|A|}  \geq \sum_{A: x\in A}  w_{f(x),A} \cdot 2^{|A|} & \geq 2^{\C_x(f)} \cdot \left( \sum_{A: x\in A}  w_{f(x),A} \right) \\
& = 2^{\C_x(f)} \enspace .
\end{align*}
Hence $\log \prt_0(f) \geq \max_{x\in f^{-1}}\{\C_x(f)\} = \C(f)$.
\item Fix  $x\in f^{-1}$. Let $b \defeq \bs_x(f)$ and let $B_1, \ldots, B_b$ be the blocks for which $f(x) \neq f(x^{B_i})$. Let $\mu_x \defeq 2^{\epsilon b-1}; \phi_x \defeq -(1-\epsilon)\mu_x$ and for each $i \in[b]$, let $-\phi_{x^{B_i}} = \mu_{x^{B_i}} \defeq \frac{2^{\epsilon b-1}}{b}; $. Let $\phi_y = \mu_y \defeq 0$ for $y \notin \{x,x^{B_1}, \ldots, x^{B_b}\}$.
\begin{enumerate}
\item Let $|A| \geq \epsilon b$. It is clear that $\forall z\in\{0,1\}^m  : \sum_{x'\in f^{-1}(z)\cap A} \mu_{x'}   + \sum_{x' \in  A} \phi_{x'}  \leq 2^{\epsilon b} \leq 2^{|A|}$.
\item Let $|A| < \epsilon b$. Let $z \neq f(x)$ or $x\notin A$. It is clear that $\sum_{x'\in f^{-1}(z)\cap A} \mu_{x'}   + \sum_{x' \in  A} \phi_{x'}  \leq 0 \leq 2^{|A|}$.
\item Let $|A| < \epsilon b$ and $z = f(x)$ and $x\in A$. Since at most $\epsilon b$ blocks among $B_1, \ldots, B_b$ can have non-empty intersection with the subset $S \subseteq [n]$ corresponding to $A$, at least $(1-\epsilon) b$ among $\{x^{B_1}, \ldots, x^{B_b}\}$ belong to $A$; therefore (since $\epsilon < 0.5$)
\begin{align*}
\sum_{x'\in f^{-1}(z)\cap A} \mu_{x'}   + \sum_{x' \in  A} \phi_{x'}  & \leq  \epsilon \cdot 2^{\epsilon b-1}  - (1-\epsilon)b \frac{2^{\epsilon b-1}}{b}  < 0 \leq  2^{|A|}.
\end{align*}
\end{enumerate}
Therefore the constraints for $\prt_{\frac{\epsilon}{4}}(f)$ are satisfied. Now,
\begin{align*}
\prt_{\frac{\epsilon}{4}}(f) \geq  \sum_{x} (1-\frac{\epsilon}{4}) \mu_{x} + \phi_{x} = (1-\frac{\epsilon}{4})2^{\epsilon b}  - (2-\epsilon) 2^{\epsilon b-1} = \epsilon 2^{\epsilon b-2} \enspace .
\end{align*}
Hence our result.
\item Let $\{w_{z,A}\}$ be an optimal solution to the primal of $\prt_\epsilon(f)$. Let $\alpha \defeq  \sum_{z} \sum_{A}  w_{z,A} \cdot 2^{|A|}$. Let $\mcA' \defeq \{ A : |A| \leq \log \frac{\alpha}{\epsilon} \}$. Then $\sum_{z} \sum_{A \notin \mcA'} w_{z,A}   \leq \epsilon$. Fix $x\in f^{-1}$. Let $\mcA'_x \defeq \{A\in\mcA' : x\in A\}$. We know that $$\alpha_x \defeq \sum_{A \in \mcA'_x} w_{f(x),A} \geq \sum_{A:x\in A} w_{f(x),A} - \epsilon \geq 1 - 2\epsilon \enspace .$$ The verifier $V_x$ for $x$ acts as follows:
\begin{enumerate}
\item Choose $A\in \mcA'_x$ with probability $\frac{w_{f(x),A}}{\alpha_x}$.
\item Query locations in $A$.
\item Accept iff locations queried are consistent with $A$. Reject otherwise.
\end{enumerate}
Now it is clear that if the input is $x$ then $V_x$ accepts with probability $1$. Also the number of queries of $V_x$ are at most $\log \frac{\alpha}{\epsilon}$ on any input $y$. Let $y$ be such that $f(y) \neq f(x)$. Let $\mcA'_{x,y} \defeq \{A\in\mcA_x' : y\in A\}$. Then,
$$ \sum_{A \in \mcA'_{x,y}} w_{f(x),A} \leq \sum_{A \in \mcA' : y \in A} \quad \sum_{z\neq f(y)} w_{z,A} \leq \sum_{A \in \mcA : y \in A} \quad \sum_{z\neq f(y)} w_{z,A} + \epsilon \leq
2\epsilon \enspace .$$
Hence $y$ would be accepted with probability at most $\frac{2\epsilon}{\alpha_x} \leq \frac{2\epsilon}{1 - 2\epsilon}$. Hence our result.

\item Let $\{w_{z,A}\}$ be an optimal solution to the primal of $\prt_{\epsilon}(f)$. Let $\alpha \defeq  \sum_{z} \sum_{A}  w_{z,A} \cdot 2^{|A|}$ and $k \defeq \log \frac{\alpha}{\epsilon}$. Let $\mcA' \defeq \{ A : |A| \leq k \}$;  then $\sum_{z} \sum_{A \notin \mcA'} w_{z,A}   \leq \epsilon$. We set $p$ as in the definition of $\cadv$ as follows. For all $x\in f^{-1}$, let $\mcA'_x \defeq \{A\in\mcA' : x\in A\}$. Define distributions $p_x$ on $[n]$ as follows:
\begin{enumerate}
\item Choose $A\in\mcA_x'$ with probability $q(x,A) \defeq \frac{w_{f(x),A}}{\sum_{A'\in \mcA'_x}w_{f(x),A'}}$.
\item Choose $i$ uniformly from the set $\{i: i \text{ appears in }A\}$.
\end{enumerate}
It is easily seen that $p_x$ is a distribution on $[n]$. We will show that
\begin{equation} \max_{x,y: f(x)\neq f(y)} \frac{1}{\sum_{i: x_i\neq y_i} \min\{p_x(i),p_y(i)\}} \leq   \frac{k}{1 - 4\epsilon}, \label{eq:plarge} \end{equation}
which proves our main claim.

Take any $x,y$ such that $ f(x) \neq f(y)$.
Let's define $\forall i \in [n] , q_x(i) \defeq \sum_{A\in \mcA'_x: i \text{ appears in } A} q(x,A)$; similarly define $q_y(i)$. It is clear that $\forall i \in [n]: q_x(i) \geq \frac{p_x(i)}{k} \text { and } q_y(i) \geq \frac{p_y(i)}{k} $. We will show:
\[\sum_{i:x_i \neq y_i} \min\{q_x(i),q_y(i)\}  \geq   1 - 4\epsilon ,\]
which implies~(\ref{eq:plarge}).

Now assume for a contradiction that $\sum_{i:x_i\neq y_i} \min\{q_x(i),q_y(i)\}<  1 - 4\epsilon.$ Consider a hybrid input $r \in \{0,1\}^n$ constructed in the following way:
if $q_x(i)\geq q_y(i)$ then $r_i \defeq x_i$, otherwise $r_i \defeq y_i$.
Now,
\begin{eqnarray*}
\sum_{A: r \in A} \quad \sum_{z}  w_{z,A} &\geq& \sum_{A \in \mcA'_r} \quad \sum_{z}  w_{z,A}  \\
& \geq & \sum_{A\in \mcA'_x} \quad  w_{f(x),A}-\sum_{i:q_x(i)<q_y(i)} q_x(i) + \sum_{A \in \mcA'_y} \quad  w_{f(y),A}  -\sum_{i:q_y(i)\leq q_x(i)} q_y(i)    \\
& \geq & \sum_{A\in \mcA'_x} \quad  w_{f(x),A} + \sum_{A \in \mcA'_y} \quad  w_{f(y),A} -\sum_{i:x_i \neq y_i} \min\{q_x(i),q_y(i)\}    \\
& \geq & \sum_{A: x\in A} \quad  w_{f(x),A} + \sum_{A: y \in A} \quad  w_{f(y),A} -\sum_{i:x_i \neq y_i} \min\{q_x(i),q_y(i)\}  - 2\epsilon  \\
&\geq& 2(1-\epsilon)- (1 - 4\epsilon) - 2\epsilon \quad > \quad  1 \enspace .
\end{eqnarray*}
This contradicts the assumption that $\{w_{z,A}\}$ is a feasible solution to the primal of $\prt_{\epsilon}(f)$.

\item Let $\{w_{z,A}\}$ be an optimal solution to the primal of $\prt_\epsilon(f)$. Let $\alpha \defeq  \sum_{z} \sum_{A}  w_{z,A} \cdot 2^{|A|}$. Let $\mcA' \defeq \{ A : |A| \leq \log \frac{\alpha}{\epsilon} \}$; then $\sum_{z} \sum_{A \notin \mcA'} w_{z,A}   \leq \epsilon$.
For $A\in \mcA'$, let $m_A(x)$ be a monomial which is $1$ iff $x \in A$.  Let $p(x) \defeq \sum_{A \in \mcA'} w_{1,A} \cdot m_A(x)$. Note that the degree of $p(x)$ is at most $\log \frac{\alpha}{\epsilon}$. Now since the constraints of the primal of $\prt_\epsilon(f)$ are satisfied by $\{w_{z,A}\}$, we get,
$$  \forall x \in f^{-1}(1): 1 \geq p(x) = \sum_{A\in\mcA': x \in A} w_{1,A} \geq \sum_{A\in\mcA: x \in A} w_{1,A} - \epsilon \geq 1 - 2\epsilon, $$
and
$$  \forall x \in f^{-1}(0): 0 \leq p(x) = \sum_{A\in\mcA': x \in A} w_{1,A} \leq \sum_{A\in\mcA: x \in A} w_{1,A} + \epsilon \leq  2\epsilon,$$
and
$$\forall x : 0 \leq p(x) \leq 1  \enspace . $$
Therefore $p(x)$, $2\epsilon$-approximates $f$ and hence our result.

\item For a Boolean function $f$, it is known that $\D(f) = O(\C(f)\bs(f))$ and $\D(f) = O(\bs(f)^3)$ (refer to~\cite{buhrman:dectreesurvey}).  The desired result is implied now using earlier parts of this theorem. 
\end{enumerate}
\end{proof}

\subsection{Example: Tribes}
In this section we give an example of applying the partition bound. We
consider the $\tribes$ function $f:\{0,1\}^n\to\{0,1\}$, which is
defined by an AND of $\sqrt n$ ORs of $\sqrt n$ variables
$x_{i,j}$. Note that $\C(f)\leq \sqrt n$, and hence $\cadv(f)\leq
O(\sqrt n)$, and that furthermore $\tdeg_{1/3}(f)$ is known to lie
between $\Omega(n^{1/3})$ and $O(\sqrt n)$. So both of the standard
general purpose lower bound methods cannot handle this problem well.

\begin{Thm}\label{thm:tribes}
Let $f$ be as above and let $\epsilon \in  (0, 1/16)$, then
\[\R_\epsilon(f)\geq \frac{1}{2} \log \prt_\epsilon(f)\geq \Omega(n).\]
\end{Thm}

\begin{proof}
We exhibit a solution to the dual of the linear program for  $\prt_\epsilon(f)$. In fact we use a one-sided relaxation of the LP for  $\prt_\epsilon(f)$, similar to the smooth rectangle bound. It is easily observed that the optimum of the LP below, denoted $\mathsf{opt}_\epsilon(f)$ is at most $\prt_\epsilon(f)$.

\vspace{0.2in}

{\footnotesize
\begin{minipage}{3in}
    \centerline{\underline{Primal}}
        \begin{align*}
      \text{min:}\quad &  \sum_{A}  w_{A} \cdot 2^{|A|} \\
       \quad &  \forall x\text{ with } f(x)=1 : \sum_{A: x \in A} w_{A} \geq 1 - \epsilon,\\
     &  \forall x\text{ with } f(x)=1 : \sum_{A: x \in A} w_{A} \leq 1 ,\\
     &  \forall x\text{ with } f(x)=0 : \sum_{A: x \in A} w_{A} \leq \epsilon ,\\
      & \forall A : w_{A} \geq 0 \enspace .
    \end{align*}
\end{minipage}
\begin{minipage}{3in}\vspace{-0.7in}
    \centerline{\underline{Dual}}
    \begin{align*}
      \text{max:}\quad &  \sum_{x:f(x)=1} (1-\epsilon) \mu_{x} -\sum_{x:f(x)=0} \epsilon \mu_x + \sum_x\phi_{x}\\
       \quad &   \forall A  : \sum_{x\in f^{-1}(1)\cap A} \mu_{x}  - \sum_{x\in f^{-1}(0)\cap A} \mu_{x} + \sum_{x \in  A} \phi_{x}  \leq 2^{|A|},\\
      & \forall x : \mu_{x} \geq 0, \phi_{x} \leq0\enspace .
    \end{align*}
\end{minipage}
}
We will work with the dual program and will assign nonzero values for $(\mu_x,\phi_x)$ on three types of inputs. Denote the set $\{(i,j):j=1,\ldots,\sqrt n\}$ by $B_i$. This is a block of inputs that feeds into a single OR. The first set of inputs has exactly one $x_{i,j}=1$ per block $B_i$. Clearly these are inputs with $f(x)=1$, and there are exactly $\sqrt n^{\sqrt n}$ such inputs. Denote the set of these inputs by $T_1$.  Then we consider a set of inputs with $f(x,y)=0$. Denote by $T_0$ the set of inputs in which all but one block $B_i$ have exactly one 1, and one block $B_i$ has no $x_{i,j}=1$. Again, there are  $\sqrt n^{\sqrt n}$ such inputs. Finally, $T_2$ contains the set of inputs, in which all $B_i$ except one have exactly one 1, and one block has two 1's. There are $(\sqrt n)^{\sqrt n}(n-\sqrt n)/2$ such inputs. Let $\delta \defeq \frac{1}{4} - 4\epsilon$ and,
\begin{eqnarray*}
\text{For all }x\in T_1 &:& \mu_x= \frac{2^{\delta n}}{\sqrt n^{\sqrt n}}; \quad \phi_x = 0,\\
\text{For all }x\in T_0&:& \mu_x= \frac{2^{\delta n}}{4\epsilon \cdot \sqrt n^{\sqrt n}}; \quad \phi_x=0,\\
\text{For all }x\in T_2&:& \phi_x= \frac{-4 \cdot 2^{\delta n}}{3 (n-\sqrt n) \sqrt n^{\sqrt n}}; \quad \mu_x=0,\\
\text{For all }x\notin T_0 \cup T_1 \cup T_2&:& \mu_x=\phi_x=0 \enspace .
\end{eqnarray*}
\begin{Claim}
$\{\mu_x,\phi_x\}$ as defined is feasible for the dual for $\mathsf{opt}_\epsilon(f)$.
\end{Claim}
\begin{proof}
Clearly $\forall x: \mu_x \geq 0, \phi_x \leq 0$.
Let $A$ be an assignment with $|A|\geq \delta n$; in this case,
\[\sum_{x\in f^{-1}(1)\cap A} \mu_{x}  - \sum_{x\in f^{-1}(0)\cap A} \mu_{x} + \sum_{x \in  A} \phi_{x}  \leq \sum_{x\in f^{-1}(1)}\mu_{x} \leq  2^{\delta n} \leq 2^{|A|}.\]
From now on $|A| <  \delta n$. Let $A$ fix at least two input positions to $1$ in a single block $B_i$. In this case clearly,
\[\sum_{x\in f^{-1}(1)\cap A} \mu_{x}  - \sum_{x\in f^{-1}(0)\cap A} \mu_{x} + \sum_{x \in  A} \phi_{x}  \leq 0 \leq 2^{|A|}.\]
Hence from now on consider $A$ which fixes at most a single input position to $1$ in each block $B_i$.
For block $i$ let  $\alpha_i$ denote the number of positions fixed to $0$ in $B_i$; let  $\beta_i\in\{0,1\}$ denote the number of positions fixed to $1$ and let $\gamma_i$ denote the number of free positions, i.e., $\sqrt n-\alpha_i-\beta_i$.

First consider the case when $k \defeq \sum_i \beta_i\leq(1-4\epsilon)\sqrt n$ and w.l.o.g. assume that the last $k$ blocks contain a 1. The number of inputs in $T_1$ consistent with $A$ is exactly $\prod_{i=1}^{\sqrt n-k} \gamma_i$. The number of inputs in $T_0$ consistent with $A$ is
 \[\sum_{i=1}^{\sqrt n-k}\prod_{j=1,\ldots,\sqrt n-k;j\neq i}\gamma_j\geq \frac{\sqrt n-k}{\sqrt n}\cdot\prod_{i=1}^{\sqrt n-k} \gamma_i\geq
 4\epsilon \prod_{i=1}^{\sqrt n-k} \gamma_i.\]
Hence,
\[\sum_{x\in f^{-1}(1)\cap A} \mu_{x}  - \sum_{x\in f^{-1}(0)\cap A} \mu_{x} + \sum_{x \in  A} \phi_{x}  \leq  \frac{2^{\delta n}}{\sqrt n^{\sqrt n}} (1 - \frac{4\epsilon}{4\epsilon} ) \prod_{i=1}^{\sqrt n-k} \gamma_i \leq 0.\]

Now assume that  $k=\sum_i \beta_i\geq(1-4\epsilon)\sqrt n$. Again w.l.o.g.~the last $k$ blocks have $\beta_i=1$. There are $\prod_{i=1}^{\sqrt n-k} \gamma_i$ inputs in $T_1\cap A$. The number of inputs in $T_2\cap A$ is at least
\[\left(\prod_{i=1}^{\sqrt n-k} \gamma_i\right)\cdot\left(\sum_{i= \sqrt n - k+1}^{\sqrt n} \gamma_i\right)\geq \left(\prod_{i=1}^{\sqrt n-k} \gamma_i\right)\cdot n(1-\delta-4\epsilon),\]
because we can choose a single $1$ for the first $\sqrt n-k$ blocks, and a second $1$ in any of the last $k$ blocks.
Hence
\begin{align*}
\sum_{x\in A\cap T_2} \phi_x & \leq -\left(\prod_{i=1}^{\sqrt n-k} \gamma_i \right) \cdot n(1-\delta-4\epsilon)\cdot \frac{4 \cdot 2^{\delta n}}{3(n-\sqrt n) \sqrt n^{\sqrt n}} \\
& = -\left(\sum_{x\in A\cap T_1} \mu_x \right) \cdot n(1-\delta-4\epsilon)\cdot \frac{4 }{3(n-\sqrt n) } \\
& \leq -\left(\sum_{x\in A\cap T_1} \mu_x \right) \cdot (1-\delta-4\epsilon) \cdot \frac{4}{3} = -\left(\sum_{x\in A\cap T_1} \mu_x \right) \enspace .
\end{align*}
Hence the constraints for dual of $\mathsf{opt}_\epsilon(f)$ are satisfied by all $A$.
\end{proof}

Finally we have,
\begin{align*}
\prt_\epsilon(f) \geq \mathsf{opt}_\epsilon(f) & \geq \sum_{x:f(x)=1} (1-\epsilon) \mu_{x} -\sum_{x:f(x)=0} \epsilon \mu_x + \sum_x\phi_{x} \\
& = 2^{\delta n}\left(1-\epsilon-\frac{\epsilon}{4\epsilon}-\frac{2}{3}\right) = 2^{\Omega(n)} \enspace .
\end{align*}
Hence our result.
\end{proof}

\subsection{Partition bound for relations}
Here we define the partition bound for query complexity for relations.
\begin{Def}[Partition Bound for relations]
Let $f \subseteq \mcX \times \mcZ$ be a relation, let $\epsilon \geq 0$. The $\epsilon$-partition bound of $f$, denoted $\prt_\epsilon(f)$, is given by the optimal value of the following linear program.

\vspace{0.2in}

{\footnotesize
\begin{minipage}{3in}
    \centerline{\underline{Primal}}\vspace{-0.2in}
    \begin{align*}
      \text{min:}\quad & \sum_{z} \sum_{A}  w_{z,A} \cdot 2^{|A|} \\
       \quad & \forall x :  \sum_{A: x \in A} \quad \sum_{z:(x,z) \in f} w_{z,A} \geq 1 - \epsilon,\\
      & \forall x : \sum_{A: x \in A} \quad \sum_{z}  w_{z,A} = 1 , \\
      & \forall z , \forall A  : w_{z,A} \geq 0 \enspace .
    \end{align*}
\end{minipage}
\begin{minipage}{3in}\vspace{-0.35in}
    \centerline{\underline{Dual}}\vspace{-0.2in}
    \begin{align*}
      \text{max:}\quad &  \sum_{x} (1-\epsilon) \mu_{x} + \phi_{x}\\
       \quad &  \forall z, \forall A : \sum_{x : x\in A; (x,z) \in f} \mu_{x}   + \sum_{x\in  A} \phi_{x}  \leq 2^{|A|},\\
      & \forall x : \mu_{x} \geq 0, \phi_{x} \in \mbR \enspace .
    \end{align*}
\end{minipage}
}
\end{Def}
As in Theorem~\ref{thm:mainquery}, we can show that partition bound is a lower bound on the randomized query complexity of $f$. Its proof is skipped since it is very similar.
\begin{Thm} Let $f \subseteq \mcX \times \mcZ$ be a relation, let $\epsilon > 0$. Then,
$\R_\epsilon(f) \geq \frac{1}{2} \log  \prt_\epsilon(f) \enspace .$
\end{Thm}

\subsection{Separations between bounds}
In this section we discuss separation between some of the bounds mentioned. 
\begin{Thm}
\begin{enumerate}
\item $\log \prt_\epsilon(\tribes)=\Omega(n)$, while $\C(\tribes), \cadv(\tribes),\adv(\tribes),\tdeg(\tribes)=O(\sqrt n)$.
\item There is a function $f:\{0,1\}^n\to\{0,1\}$ such that $\log \prt_\epsilon(f)\geq\Omega(\bs(f))=\Omega(n)$, while $\edeg(f)=O(n^{0.62})$ for all $\epsilon<1/2$.
\end{enumerate} \label{thm:sepq}
\end{Thm}
\begin{proof}
\begin{enumerate}
\item The lower bound $\log \prt_\epsilon(\tribes)=\Omega(n)$ is shown in Theorem~\ref{thm:tribes}. The upper bound on $\C(\tribes)$ is obvious, and implies the bound on $\cadv$. The remaining bound follow from the existence of efficient quantum query algorithms for the problem.
\item Examples of such functions are given in \cite{nisan&szegedy:degree,nisan&wigderson:rank} with the best construction attributed to Kushilevitz in the latter paper.
\end{enumerate}
\end{proof}

In the first result above the partition bound with error beats all of the "standard" lower bound methods for randomized query complexity (as well as $\C$). In the second result the partition bound is better than the exact degree. By composing $\tribes$ with the function $f$ above we can also get a function for which $\log\prt_\epsilon$ is polynomially larger than $\C$ and $\edeg$ simultaneously.

\subsection{Boosting}
We remark, without proof, that the error in the partition bound (both communication and query) and its relatives can in general be boosted down in the same way as the error for randomized protocols, for example we have: For all relations $f$: $\log\prt_{2^{-k}}(f)=O(k\cdot \log \prt_{1/3}(f))$.

\section{Open Questions} Here we state some of the questions left open.

\vspace{0.2in}

\noindent {\bf Communication Complexity}
\begin{enumerate}
\item Is $\R_{1/3}(f) = \poly(\log \prt_{1/3}(f))$ for all relations $f$?
\item Is $\prt_{1/3}(\tribes) = \Omega(n)$ ?
\end{enumerate}

\noindent {\bf Query Complexity}
\begin{enumerate}\item  Is $\R_{1/3}(f) = O(\log^2 \prt_{1/3}(f))$ or better still is $\R_{1/3}(f) = O(\log \prt_{1/3}(f))$ ? 
\item Is $ \adv(f) = O(\log \prt_{1/3}(f))$ ?
\item    Is  $\edeg(f) = \widetilde{O}(\prt_0(f))$ ?
\end{enumerate}

\subsection*{Acknowledgment} We thank Shengyu Zhang for helpful discussions. The work done is supported by the internal grants of the Centre for Quantum Technologies, Singapore.

\newcommand{\etalchar}[1]{$^{#1}$}

\end{document}